\documentclass[a4paper,aps,prd,twocolumn,tightenlines,preprintnumbers,nofootinbib,showkeys,superscriptaddress,longbibliography,notitlepage]{revtex4-1}
\pdfoutput=1

\usepackage[dvipsnames]{xcolor}
\usepackage[export]{adjustbox}
\usepackage{tikz,contour} 
\usetikzlibrary{shapes,arrows,positioning,automata,backgrounds,calc,er,patterns}
\usepackage{fullpage}
\usepackage{relsize}
\usepackage[left=1.6cm,right=1.6cm,top=1.6cm,bottom=1.65cm]{geometry}
\usepackage{multirow}
\usepackage{booktabs}
\usepackage{XCharter}
\usepackage{mathptmx}
\usepackage[T1]{fontenc}
\usepackage{pifont}
\usepackage{fix-cm}
\usepackage{amsfonts,amsmath,amsthm,amssymb,amsbsy,bm,mathtools,latexsym,esvect}
\usepackage{cancel}
\usepackage{enumitem}
\usepackage{hyperref}
\hypersetup{colorlinks=true,citecolor=red,linkcolor=NavyBlue,urlcolor=NavyBlue}
\usepackage[caption=false,labelformat=simple]{subfig}
 
\usepackage{natbib}
\usepackage{physics} 
\usepackage{graphicx}

\newtheorem{theorem}{Theorem}

\begin{document}
\title{PINNing the pion: conformal deep learning for $F_\pi(s)$ and the $(g-2)_\mu$ hadronic contribution}
%\title{PINNing the pion: A first-principles approach to the pion form factor and the $(g-2)_\mu$ tension}

\author{Mayank Goel}
\email{mayank.goel447@gmail.com}
\affiliation{Center for Computational Natural Sciences and Bioinformatics, International Institute of Information Technology, Hyderabad 500 032, India}

\author{Subhadip Mitra}
\email{subhadip.mitra@iiit.ac.in}
\affiliation{Center for Computational Natural Sciences and Bioinformatics, International Institute of Information Technology, Hyderabad 500 032, India}
\affiliation{Center for Quantum Science and Technology, International Institute of Information Technology, Hyderabad 500 032, India}

\author{Monalisa Patra}
\email{monalisa.p@ihub-data.iiit.ac.in}
\affiliation{iHub-Data, International Institute of Information Technology, Hyderabad 500 032, India}

%%%%%%%%%%%%%%%%%%%%%%%%%%%%%%%%%%%%%%%%%%%%%%%%%%%%%%%
\begin{abstract}
\noindent
Extracting the pion electromagnetic form factor $F_{\pi}(s)$ through phenomenological curve-fitting models introduces model dependence, unphysical artefacts, and kinematic inconsistencies. We introduce a Physics-Informed Neural Network (PINN) embedded in a conformal $z$-plane that constructs $F_{\pi}(s)$ directly from first principles across spacelike and timelike domains: charge normalisation and Schwarz reflection are enforced by construction, while Cauchy-Riemann analyticity, dispersion relations, Watson's theorem, and perturbative QCD asymptotics enter through the loss functional. Thus, the fundamental S-matrix principles dictate the form factor's behaviour while data act as constraints. Mapping the cut complex plane onto the unit disk bounds the Hessian norm and prevents Neural Tangent Kernel spectral starvation, two known failure modes of deep-learning optimisation. Besides $e^+e^-$ scattering data, we also incorporate $\tau$-decay data through a switch that isolates the pure isovector form factor natively, bypassing model-dependent isospin-breaking pre-corrections. The network organically yields an interior zero-free form factor, while the framework tests experimental tensions around the $\rho(770)$ peak against analyticity and dispersion constraints. We obtain model-independent estimates of the pion charge radius, $\langle r_{\pi}^2 \rangle = 0.435 \pm 0.008_{\text{stat}} \pm 0.007_{\text{cali}}$ fm$^2$, the second-sheet pole parameters, $m_{\rho}^{\text{pole}} = 761.72\pm 1.04$ MeV and $\Gamma_{\rho}^{\text{pole}} = 135.99 \pm 1.20$ MeV, and the two-pion contribution to the muon anomalous magnetic moment, $a_{\mu}^{\pi\pi} = (506.48 \pm 2.02_{\text{stat}} \pm 1.70_{\text{cali}}) \times 10^{-10}$.
\end{abstract}
%%%%%%%%%%%%%%%%%%%%%%%%%%%%%%%%%%%%%%%%%%%%%%%%%%%%%%%

\maketitle
%%%%%%%%%%%%%%%%%%%%%%%%%%%%%%%%%%%%%%%%%%%%%%%%%%%%%%%
\section{Introduction}
\noindent
The pion electromagnetic form factor $F_\pi(q^2)$ is a fundamental quantity in hadron physics. It maps the spatial distribution of valence quarks within the lightest hadron. As the pion is the pseudo-Goldstone boson of broken chiral symmetry, its structure provides a direct window into the transition between perturbative QCD at high energies and the QCD vacuum. It also sits at the heart of the muon $g-2$ anomaly because the dominant two-pion hadronic vacuum polarisation (HVP) contribution to $a_{\mu}=(g-2)_\mu/2$ is dictated by an integral over $|F_{\pi}(q^{2})|^{2}$. The historical $3-4\sigma$ tension with the Standard Model (SM) is fundamentally tied to uncertainties in this form factor~\cite{Muong-2:2006rrc, Muong-2:2021ojo}. However, recent high-precision lattice QCD calculations~\cite{Borsanyi:2020mff} and the CMD-3 cross-section measurement~\cite{CMD-3:2023rfe} have upended this landscape by suggesting a higher HVP value that largely defuses the $g-2$ discrepancy. But in the process, they have triggered a new crisis: a contradiction between theory and the traditional data-driven consensus. Resolving this internal HVP puzzle now hinges entirely on how reliably $F_{\pi}(q^{2})$ can be continuously mapped across both the spacelike and timelike domains.

Several theoretical approaches to study $F_\pi(q^2)$ exist in the literature, but these usually struggle to bridge the domains while satisfying all constraints simultaneously. Chiral perturbation theory (ChPT) is precise near $q^{2} = 0$~\cite{Gasser:1983yg, *Gasser:1984gg} but breaks down approaching the $\rho(770)$ resonance region because it lacks the vector meson degrees of freedom. Vector meson dominance (VMD) models~\cite{Gounaris:1968mw} capture this resonance peak but violate the $1/q^{2}$ power-law scaling at high energies required by perturbative quantum chromodynamics (pQCD) without any connection to underlying quark-gluon dynamics. Dispersive frameworks based on the Omn\`es representation~\cite{Omnes:1958hv} rely on experimental inputs for $\pi\pi$ scattering phase shifts obtained via the Roy equations~\cite{Roy:1971tc, Colangelo:2001df}. These methods accumulate systematic uncertainties, as Leplumey and Stoffer recently identified, depending on how they treat complex zeros in the form factor~\cite{Leplumey:2025kvv}. Imposing a zero-free condition reduces these systematic errors but rigidifies the model, exposing severe contradictions between the CMD-3~\cite{CMD-3:2023rfe} and BaBar~\cite{BaBar:2009wpw, BaBar:2012bdw} datasets. Lattice QCD offers a first-principles alternative but remains confined to the Euclidean region, where continuing these numerical results to Minkowski space remains limited by the incompleteness problem~\cite{RuizArriola:2026qiw}.

Traditional fitting procedures assume some physics-motivated functional form of the form factor and thus suffer from modelling bias. Modern machine learning and data-driven interpolations~\cite{BuarqueFranzosi:2021wrv} can bypass explicit model parameterisation, but they do not guarantee global S-matrix consistency and are susceptible to unphysical numerical artefacts. Moreover, as we demonstrate here, standard neural network fits face a structural obstruction: training directly with the physical momentum-transfer variable $s$ ($\equiv q^{2}$) causes gradient descent to fail beyond a region. The unbounded kinematic domain forces the minimum eigenvalue of the neural tangent kernel (NTK), the functional matrix governing the optimisation dynamics in the network's parameter space, to collapse to zero. This structural pathology causes spectral starvation, meaning the network fails to learn the form factor in the high-$q^{2}$ region, rendering the asymptotic high-energy constraints unlearnable regardless of the network architecture or training time.

The resolution of this inherent obstruction reveals a correspondence between S-matrix theory and deep learning optimisation. The conformal mapping $z(s)$ used in dispersive physics (see, e.g., ~\cite{Caprini:1999ws}) to compactify the cut complex plane and enforce analyticity matches the transformation required to stabilise gradient descent and preserve the NTK expressivity. Compactifying the physical domain into the unit disk simultaneously satisfies the topological demands of the S-matrix and eliminates the coordinate-induced gradient starvation in the optimisation space. The same geometry addresses the physics and the optimisation.

We use this geometric transformation to build a physics-informed neural network (PINN) that operates natively in the conformal $\mathcal{Z}$ plane and constructs the form factor from the data by relying on the \emph{S-matrix principles} rather than simple data fitting. Our architecture structurally enforces normalisation and the Schwarz reflection principle, while loss functions enforce analyticity, the dispersion relation, Watson's theorem, and perturbative QCD asymptotics. We don't fix any specific functional form for the imaginary part $\text{Im}\,F_{\pi}$. When trained on spacelike and timelike data simultaneously, the reconstructed form factor emerges free of complex zeros in the interior. Because this zero-free structure is organic rather than imposed, the framework retains the flexibility to accommodate tensions in the data. The result is a stable, model-independent reconstruction of $F_{\pi}(q^{2})$ across all kinematic regions, providing independent evidence for a zero-free form factor. From the numerical form factor, we estimate the pion charge radius, unphysical second-sheet resonance pole parameters, and the two-pion HVP contribution to $a_\mu$. 

The rest of this paper is organised as follows. Section~\ref{sec:theory} presents the theoretical background and the fundamental principles governing the form factor. Section~\ref{sec:nn} establishes the conformal $z$ mapping, outlines the PINN architecture, including the details of the composite loss functional that enforces physics constraints, and explains how conformal preconditioning resolves deep-learning optimisation pathologies. Section~\ref{sec:expt-data} describes the experimental datasets covering spacelike and timelike channels used to train and constrain the model. In Section~\ref{sec:results}, we present our primary results, including the extracted pion form factor, the network's role as an analytical arbiter for dataset tensions around the $\rho(770)$ peak, and precise extractions of the pion charge radius, resonance pole parameters, and $a_{\mu}^{\pi\pi}$. Section~\ref{sec:pade} provides a compact conformal Pad\'e approximant that compresses the continuous network prediction into a usable closed-form expression for phenomenological applications. Finally, Section~\ref{sec:discussion} summarises our findings, discusses the framework's main limitations, and outlines potential future directions.

%%%%%%%%%%%%%%%%%%%%%%%%%%%%%%%%%%%%%%%%%%%%%%%%%%%%%%%
\section{The theoretical background} \label{sec:theory}
%%%%%%%%%%%%%%%%%%%%%%%%%%%%%%%%%%%%%%%%%%%%%%%%%%%%%%%
\subsection{The low-energy expansion}
\noindent
The pion form factor can be formally defined via the matrix element of the electromagnetic current operator $J^\mu_{\text{em}}$ evaluated between charged pion states:
\begin{equation}
\langle \pi^+(p') | J^\mu_{\text{em}}(0) | \pi^+(p) \rangle = (p + p')^\mu F_\pi(s),
\end{equation}
where $s \equiv q^2 = (p' - p)^2$ is the squared momentum transfer.  The conservation of the vector current associated with the unbroken $SU(2)_V$ isospin subgroup of the global chiral symmetry demands the charge normalisation $F_{\pi}(0) = 1$. The low-energy expansion of the form factor serves as a fundamental probe of the QCD vacuum:
\begin{align}
    F_{\pi}(s) = 1 + \frac{1}{3!}\langle r_{\pi}^2 \rangle s + \frac{1}{5!}\langle r_{\pi}^4 \rangle s^2 + \frac{1}{7!}\langle r_{\pi}^6 \rangle s^3 + \mathcal{O}(s^4),\label{eq:charge_radius}    
\end{align}
where $\langle r_{\pi}^2 \rangle$ is the charge radius of the pion, and the rest captures the higher-order curvature. Any physically admissible parameterisation must reproduce both the unit charge normalisation and the observed slope at the origin.

%%%%%%%%%%%%%%%%%%%%%%%%%%%%%%%%%%%%%%%%%%%%%%%%%%%%%%%
\subsection{Analyticity and the S matrix}
\noindent
Causality and the analytic structure of the S matrix mandate that $F_\pi(s)$ is a holomorphic function throughout the complex $s$ plane except for the branch cuts and branch points. Because physical signals cannot propagate outside the forward light cone, the commutators of local currents vanish for spacelike separations. In momentum space, this microcausality demands that the singularities of the form factor are strictly confined to regions corresponding to physical intermediate hadronic states. The lightest such state in the isovector vector channel is the two-pion state, producing a branch cut along the positive real axis beginning at $s_{\mathrm{th}} = 4m_\pi^2$. 

The Schwarz reflection principle, $F_\pi(s^*) = F_\pi^*(s)$ follows from analyticity. It implies that \(F_\pi(s)\) is real on the real axis wherever the function is analytic, in particular in the spacelike region and the unphysical region $0<s<4m_\pi^2$. On the physical cut, the upper and lower rim values are related by
\begin{equation}
    F_\pi(s-i0)=F_\pi^*(s+i0),
    \quad s>4m_\pi^2,
\end{equation}
and the discontinuity is therefore
\begin{equation}
    \operatorname{Disc}F_\pi(s) = F_\pi(s+i0)-F_\pi(s-i0) = 2i\,\operatorname{Im}F_\pi(s+i0).
\end{equation}
Therefore, the absorptive part of the form factor completely determines the discontinuity across the physical branch cut. 

%%%%%%%%%%%%%%%%%%%%%%%%%%%%%%%%%%%%%%%%%%%%%%%%%%%%%%%
\subsection{Isospin breaking, the inelastic threshold, and Watson's theorem}
\noindent
Below the inelastic four-pion threshold ($s_{\rm th}^{\rm inelastic} = 16m_\pi^2\approx 0.31$ GeV$^2$), the final state in $e^+e^-$ scattering is dominated by the elastic $\pi\pi$ channel. Since the two-pion state is an isovector state, unitarity implies that Watson's final-state interaction theorem is applicable in this region. The theorem fixes the phase of the pion form factor to the isospin-$1$, $P$-wave $\pi\pi$ scattering phase shift,
\begin{equation}
    \arg \left[F^{I=1}_\pi(s+i0)\right] = \delta_1^1(s), \quad 4m_\pi^2 < s < s_{\rm th}^{\mathrm{inelastic}}.\label{eq:watson_phase}
\end{equation}
Equivalently, for isovector dominance,
\begin{equation}
    \operatorname{Im}F_\pi(s) = \operatorname{Re}F_\pi(s)\tan\delta_1^1(s).\label{eq:watson_relation}
\end{equation}
This condition strongly constrains the phase of $F_\pi(s)$. 

In reality, isovector dominance continues much beyond $16m_\pi^2$. In particular, the rapid increase of $\delta_1^1(s)$ through approximately $\pi/2$ in the $\rho(770)$ region forces the form factor to exhibit the characteristic resonant enhancement observed in the timelike cross section. Beyond the $\rho$ resonance, electromagnetic interactions induce small isospin-breaking mixing between the $I=1$ ($\rho$) and $I=0$ (mostly $\omega$, but also the heavier $\phi$) states that manifests as minor (compared to the isovector resonance) dip-bump interference patterns in the timelike cross-section data. In the $e^+e^- \to \pi^+\pi^-$ process, this isospin breaking is quantified by $\rho$--$\omega$ mixing. The physical electromagnetic form factor is parameterised as
\begin{equation}
    F_\pi^{e^+e^-}(s) = F_\pi^{I=1}(s)\left(1 + \frac{\alpha_{\rho\text{-}\omega}\, e^{i\phi_\omega} m_\omega^2}{m_\omega^2 - s - im_\omega\Gamma_\omega}\right),\label{eq:mixing}
\end{equation}
where $\alpha_{\rho\text{-}\omega}$ controls the mixing strength and $\phi_\omega \approx 1.74$ rad is the Orsay phase, producing the characteristic interference near $s \approx m_\omega^2$. In contrast, $\tau$-decay data measure $F_\pi^{I=1}(s)$ directly, free of this contamination up to calculable isospin-breaking corrections. 

The first physically dominant inelastic channel breaking the elastic phase-locking is $\omega\pi^0$ production. This shifts the effective phenomenological inelastic threshold to $s_{\rm eff.}^{\rm inelastic} = (m_{\omega} + m_{\pi})^2 \approx 0.84\text{ GeV}^2$~\cite{Ananthanarayan:2018nyx}. However, since even this effect is small, the effective inelastic threshold is generally set at the $K\overline K$ threshold at $4m_K^2\approx 1$ GeV$^2$~\cite{Ananthanarayan:2000ht}. We therefore consider $s_{\rm eff.}^{\rm inelastic} \approx 1$ GeV$^2$ as the upper boundary of the elastic Watson constraint in Eq.~\eqref{eq:watson_relation} in the loss functional of our neural network.

%%%%%%%%%%%%%%%%%%%%%%%%%%%%%%%%%%%%%%%%%%%%%%%%%%%%%%%
\subsection{Asymptotic pQCD scaling}
\noindent
At large spacelike momentum transfers, perturbative QCD (pQCD) imposes an absolute asymptotic boundary. In the deep Euclidean region ($s \to -\infty$), the virtual photon resolves the pion's partonic structure. Here, the Brodsky-Farrar quark-counting rules~\cite{Brodsky:1973kr,*Brodsky:1974vy} dictate that the form factor must decouple, scaling as $F_\pi(s) \sim 1/s$ up to logarithmic corrections from the running strong coupling and the pion distribution amplitude~\cite{Melic:1998qr,*Melic:1999mx}:
\begin{equation}
    F_\pi(-Q^2) \sim \frac{\mathcal A}{Q^2}, \quad Q^2=-s \to \infty, \label{eq:pqcd_ff}
\end{equation}
where
\begin{align*}
    \mathcal A \approx 8\pi f_{\pi}^2 \alpha_S(\mu_R^2) \Bigg[ 1 + \frac{\alpha_S(\mu_R^2)}{\pi} \Big( \frac{\beta_0}{4} \ln\left(\frac{\mu_R^2}{Q^2}\right) 
 + 6.41 \Big) \Bigg],
\end{align*}
at the next-to-leading order (NLO) in pQCD with $\mu_R^2\approx Q^2/21$ being the renormalisation scale (see Appendix~\ref{sec:nlo}), and  $f_\pi=0.131$ GeV is the pion decay constant. This condition forbids any parametrisation that grows unboundedly at high energies, ensuring consistency with the short-distance limits of QCD.\bigskip

\noindent
Traditionally, accommodating all of unit normalisation, phase tracking, analytic continuity, and asymptotic scaling within a single phenomenological model has proven remarkably difficult. However, our framework does not need any form factor model to enforce these S-matrix and QCD mandates; instead, these enter via a loss functional that governs the neural network's optimisation.

%%%%%%%%%%%%%%%%%%%%%%%%%%%%%%%%%%%%%%%%%%%%%%%%%%%%%%%
\begin{figure*}
    \centering
    \includegraphics[width=\textwidth]{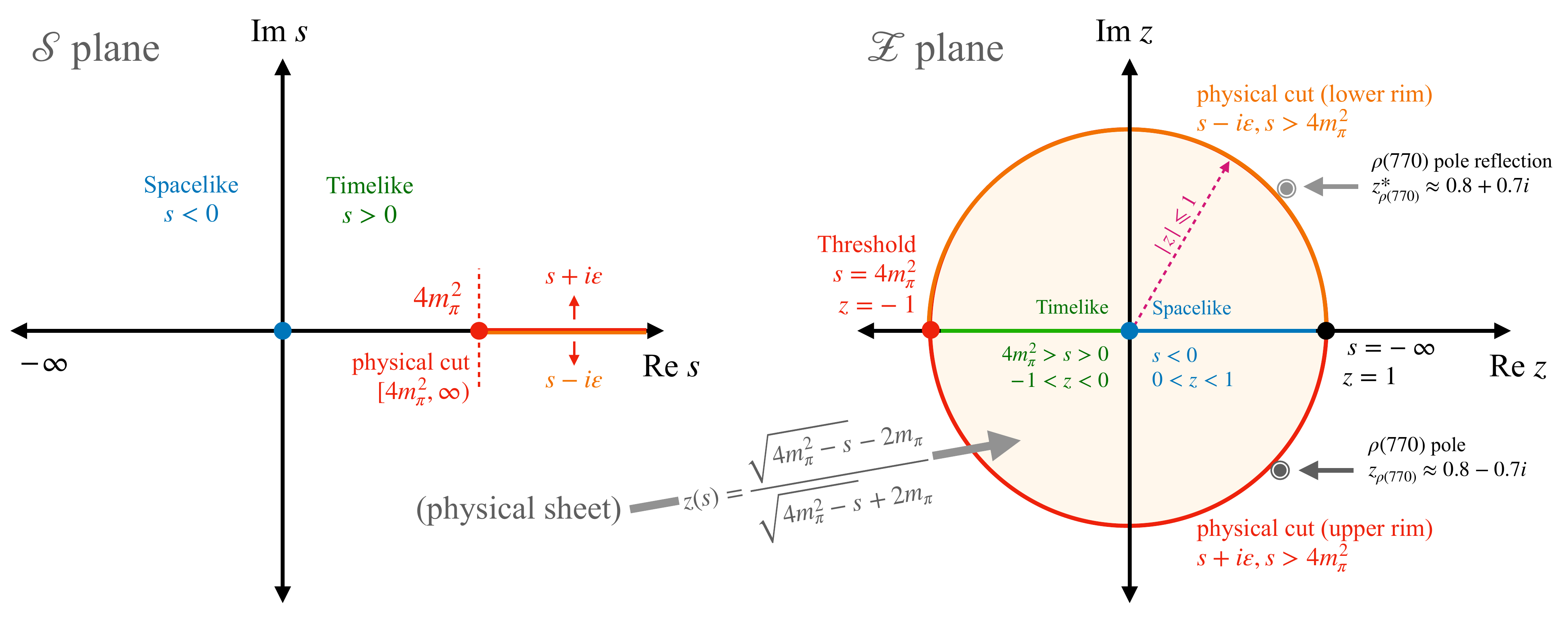}
    \caption{The conformal map $z(s)$ compactifying the physical $\mathcal{S}$ plane onto the unit disk. The physical branch cut $[4m_\pi^2, \infty)$ maps to the unit circle boundary. The $\rho(770)$ resonance pole lies on the second Riemann sheet, outside the unit disk at $z_{\rho(770)} \approx 0.8 - 0.7i$, ensuring $F_\pi$ is holomorphic throughout the interior $|z| < 1$. \label{fig:CMap}}
\end{figure*}
%%%%%%%%%%%%%%%%%%%%%%%%%%%%%%%%%%%%%%%%%%%%%%%%%%%%%%%

%%%%%%%%%%%%%%%%%%%%%%%%%%%%%%%%%%%%%%%%%%%%%%%%%%%%%%%
\section{Neural network architecture: Stability and expressivity}\label{sec:nn}
%%%%%%%%%%%%%%%%%%%%%%%%%%%%%%%%%%%%%%%%%%%%%%%%%%%%%%%
\subsection{Using conformal mapping to bypass geometric instability}
\noindent
For a neural network, the primary numerical challenge comes from unbounded physical kinematic variables, which may contain several nontrivial boundary components associated with different thresholds. Hadronic form factors possess several threshold singularities associated with distinct physical channels. These thresholds generate branch points, and the corresponding branch cuts determine the boundary structure of the physical sheet. Since the physical momentum of newly created massive particles scales as $\sqrt{s - s_{\text{th}}}$, each independent threshold locally splits the complex $s$-plane into a two-sheeted Riemann surface. Furthermore, if massless exchanges are present, they introduce logarithmic singularities that splinter the domain into infinitely many sheets. To address this, we first replace this complicated physical geometry with a bounded canonical domain:
\begin{equation}
\mathcal{S} = \hat{\mathbb{C}} \setminus \bigcup_{i=1}^{N_\Gamma}\Gamma_i,
\end{equation}
where $\hat{\mathbb{C}} = \mathbb{C}\cup\{\infty\}$ is the Riemann sphere and $\Gamma_i$'s are disjoint branch cuts on the physical sheet. As the cuts define non-degenerate boundary components, $\mathcal{S}$ is a finitely connected planar domain. However, as we discuss below, even the canonical physical domain is not ideal for training our neural network.

For optimisers such as Adam or L-BFGS, optimisation is most stable when the local loss landscape is close to isotropic, i.e., approximately bowl-shaped with comparable curvature in all directions. This is hard to achieve everywhere in the physical space. However, a transformation to a conformal coordinate can act as a geometric preconditioner and remove the coordinate-induced source of large curvature associated with the unbounded physical momentum-transfer variable. To see that, let us first define the loss function. Let $f(x;\theta) = \sum_{k=1}^{K} v_k \sigma(a_k x+b_k)$ be a single-hidden-layer neural network with $K$ neurons, $P=3K$ parameters: $\theta = \{v_k,a_k,b_k\}_{k=1}^{K}$ (where $x$ denotes the input coordinate), and $\sigma$ a smooth and bounded activation function (i.e., $|\sigma(u)| \leqslant C_\sigma$ and  $|\sigma'(u)| \leqslant C_{\sigma'}, \forall u\in\mathbb{R}$). The empirical mean-squared-error loss over $N$ collocation points is given as
\begin{equation}
\mathcal{L}(\theta) = \frac{1}{2N} \sum_{i=1}^{N}
    \left(f(x_i;\theta)-y_i\right)^2 = \frac{1}{2N} \sum_{i=1}^{N} r_i^2(\theta),
\end{equation}
where $r_i(\theta) = f(x_i;\theta)-y_i$ denotes the $i$th residual. The Hessian matrix of the loss, $H\in \mathbb R^{P\times P}$, is the collection of all second derivatives, given as
\begin{align}
H_{jk} =&\ \partial_{\theta_j }\partial_{\theta_k } \mathcal{L}\nonumber\\
=&\ \frac{1}{N} \sum_{i=1}^{N} \left[ \partial_{\theta_j} f(x_i; \theta) \partial_{\theta_k} f(x_i; \theta)  + r_i(\theta)\, \partial_{\theta_j }\partial_{\theta_k }f(x_i; \theta) \right] \nonumber\\   
=&\ \frac{1}{N} \left[ (J^T J)_{jk} + \sum_{i=1}^{N}  r_i(\theta)\, \partial_{\theta_j }\partial_{\theta_k }f(x_i; \theta) \right],\label{eq:hessian}
\end{align}
where $J\in\mathbb{R}^{N\times P}$ is the Jacobian, $J_{ij}=\partial_{\theta_j} f(x_i;\theta)$.

Since near a small-residual solution ($r_i(\theta)\approx 0$) the Hessian is approximated by the Gauss-Newton matrix, $H_{\mathrm{GN}} = (J^TJ)/N$, we can extract a limit on the largest eigenvalue: $ \lambda_{\max} (H_{\mathrm{GN}}) \leqslant \operatorname{Tr}(H_{\mathrm{GN}}) =  \|J\|_F^2/N$. This upper bound on the eigenvalue dictates the maximum curvature of our optimisation landscape. It also points to the crucial dependence on the choice of coordinate system: if the network is trained directly in physical momentum space, this curvature can become unbounded. However, if we map the problem into a conformal space, it need not be. Let's analyse the behaviour of the network gradient in these two spaces.

In the physical domain, the input coordinate, $x = s \in(-\infty,s_{\mathrm{th}}=4m_\pi^2] $. The deep spacelike region extends to negative infinity ($s \to -\infty$). If we calculate the derivative of the network's prediction with respect to an internal frequency weight $a_k$ (assuming its corresponding outer weight $v_k \neq 0$), the chain rule yields: $\partial_{a_k} f(s;\theta) = v_k \sigma'(a_k s + b_k) s$. As we evaluate the loss using data from the deep spacelike region, the magnitude of $s$ grows unboundedly. Provided the activation function's slope does not vanish at these extremes (which is generically true for the periodic sine activations, $|\sigma'| \geqslant c_{\sigma} > 0$), this gradient can diverge uncontrollably, forcing the Frobenius norm $\|J\|_F^2$ and thus $\lambda_{\max}(H_{\mathrm{GN}})$ to explode.  The gradient descent becomes inherently unstable due to steep ravines in the loss landscape. This unbounded curvature reveals a critical vulnerability in standard neural network architectures. If we attempt to avoid this explosion by using traditional activation functions like $\tanh$, their exponentially decaying derivatives ($\sigma'(u) = \operatorname{sech}^2(u)$) lead to vanishing gradients and spectral starvation, which means the network does not learn the form factor in the high $q^2$ region. Conversely, highly expressive periodic activations, like the adaptive sine~\cite{sitzmann2020siren}, where $|\sigma'| \geqslant c_{\sigma} > 0$, preserve gradient signals but trigger the uncontrolled explosion in the unbounded $\mathcal{S}$ space. 

Now, instead of operating in the physical $\mathcal{S}$-space, let us consider a conformal bijection into the complex $\mathcal{Z}$-plane, where the mapping $z(s)$ compactifies the entire unbounded physical domain ($\mathcal{S}$-space) into a finite-sized one, like, e.g., the unit disk, $|z| \leqslant 1$. For hadronic form factors in general, we can define this bounded domain generally as 
\begin{equation}
\mathcal{Z} = \mathbb{D}\setminus\bigcup_{j=1}^{N_{\Gamma}-1}\overline{D_j(c_j,r_j)},
\end{equation}
where the closed disks are mutually disjoint and contained in $\mathbb{D}= \{z\in\mathbb{C}: |z|<1\}$. However, for the pion form factor, it simplifies beautifully. Since there is only one primary branch cut starting at $4m_\pi^2$, the conformal bijection~\cite{Caprini:1999ws}
\begin{align}
z(s) = \frac{\sqrt{4m_\pi^2 - s} - 2m_\pi}{\sqrt{4m_\pi^2 - s} + 2m_\pi}\label{eq:confmap}    
\end{align} 
smoothly compactifies the entire first Riemann sheet directly into the simple unit disk $|z| \leqslant 1$ (see Fig.~\ref{fig:CMap}).

Let our network in this new conformal domain be parametrised as: $f(z;\theta) = \sum_{k=1}^{K} v_k \sigma(a_k z + b_k)$. Because $z$ is strictly bounded, the catastrophic gradient growth we observed earlier is mathematically eliminated, as the activation derivative $|\sigma'|$ is already capped. As long as the network's weights remain finite during optimisation, the magnitude of the gradient is guaranteed to be bounded by a finite constant. We can extend this logic to the second derivatives. Every term within the exact Hessian matrix relies solely on products of inherently bounded quantities such as the bounded input coordinate $z$, the finite network weights, bounded experimental data targets ($|y_i| \leqslant Y$), and the bounded activation derivatives ($\sigma'$ and $\sigma''$). Because no single element can diverge, there must exist a finite constant $C_H$ such that the Hessian of the loss is bounded (see Appendix~\ref{sec:thoremproofs} for the proof): 
\begin{align}
\|H\|_2 \leqslant C_H < \infty.   
\end{align}

The bound depends on the network size and weight bounds, but not on the original unbounded physical coordinate $s$. By acting as a geometric preconditioner, the conformal mapping ensures the loss landscape possesses a well-defined maximum curvature. This implies that the loss gradient is globally Lipschitz continuous with a Lipschitz constant (maximum rate of change) $L = C_H$. Effectively, this establishes a curvature speed limit for the network's optimisation dynamics. According to the standard descent lemma~\cite{Nesterov2014IntroductoryLO}, taking a step along the gradient yields the following bound on the new loss:
\begin{equation}
\mathcal{L}(\theta-\eta \partial_\theta\mathcal{L}(\theta)) \leqslant \mathcal{L}(\theta) - \eta \left(1-\frac{\eta L}{2}\right)
    \|\partial_\theta\mathcal{L}(\theta)\|_2^2 .
\end{equation}
The bracketed term on the right remains positive as long as the learning rate satisfies $0 < \eta < \frac{2}{L}$. Therefore, the loss decreases at every step unless a stationary point ($\partial_\theta\mathcal{L}(\theta)=0$) is reached. In addition to bypassing the instabilities that plague the physical momentum space, this approach ensures robust convergence for gradient-descent-based optimisation.

For the neural network, the crucial point is simply that its inputs remain bounded. Unlike traditional dispersive approaches, where the choice of conformal map encodes physical information about the branch cut structure, our framework is agnostic to this choice: any bijection that compactifies the unbounded physical domain into a finite region -- a disk, an ellipse, or any smoothly bounded set -- would eliminate the catastrophic gradient explosions and guarantee Hessian boundedness.\footnote{This is perhaps a curious feature of our physics-informed approach. While the PINN enforces the physical constraints of the S-matrix, like analyticity, unitarity, etc., through its loss functional, any conformal bijection to a smoothly bounded set would work; the physics enters through the constraints, not this choice of coordinate. } We adopt the map in Eq.~\eqref{eq:confmap}, well-known in the literature~\cite{Boyd:1994tt}, as a convenient and physically motivated choice.

%%%%%%%%%%%%%%%%%%%%%%%%%%%%%%%%%%%%%%%%%%%%%%%%%%%%%%%
\subsection{Preserving the Neural Tangent Kernel}
\noindent
In addition to successful convergence, we also require that the network maintain the expressivity needed to fit the data. We examine the network's training dynamics in the parameter space through the empirical NTK~\cite{Jacot:2018ivh}, defined as:
\begin{equation}
\Theta = \frac{1}{P} JJ^T \ \in \mathbb{R}^{N \times N}.
\end{equation}
The diagonal elements of this kernel represent the local gradient magnitude for a given data point: $\Theta_{ii} =  \|\partial_\theta f(x_i;\theta)\|_2^2/P$. For a continuous-time gradient descent with the mean-squared error (MSE) loss, the residual dynamics are governed by (see the appendix for proof):
\begin{equation}
    \frac{1}{\mathcal{L}}\frac{d\mathcal{L}}{dt} \leqslant -\frac{2 P}{N} \lambda_{\min}(\Theta).
\end{equation}
Thus, the NTK must remain positive definite ($\lambda_{\min}(\Theta) > 0$) for $\mathcal L$ to converge to a global minimum successfully. If the minimum eigenvalue approaches zero, the network undergoes \emph{spectral starvation} and loses the ability to learn the corresponding eigen-directions of the dataset.

This exposes another limitation of using the physical coordinates as direct inputs. The physical form factor must vanish in the deep Euclidean limit ($F_\pi(s) \to 0$ as $s \to -\infty$). If the network is trained in the unbounded $\mathcal{S}$-space using a periodic activation such as the adaptive sine, $\sigma(a_ks+b_k)$ has no limit as $s\to-\infty$ for any fixed $a_k\neq0$, as it oscillates indefinitely rather than settling to a fixed value. The only way for the network to approach a finite, tunable value in this limit is for the optimiser to shrink the internal weights ($a_k \to 0$). Hence, the parameter gradients evaluated at these deep spacelike points vanish:
\begin{equation}
    \lim_{s_i \to -\infty} \|\partial_\theta f(s_i;\theta)\|_2^2 = 0 \implies \Theta_{ii} \to 0.
\end{equation}
Because the minimum eigenvalue of a positive semi-definite matrix is upper-bounded by its minimum diagonal element, evaluating the network in the unbounded domain forces $\lim_{s_i \to -\infty} \lambda_{\min}(\Theta_s) = 0$. Thus, the diminishing form factor induces spectral starvation, flattening the loss landscape and stalling the optimisation.

Again, mapping the domain to the conformal $\mathcal{Z}$-space eradicates this null space. Because the conformal coordinate maps the infinite spacelike tail to a finite boundary point ($z \to 1$), the network parameters are no longer forced to collapse to zero to satisfy the asymptotic behaviour. For generic, non-polynomial activations like the adaptive sine function, the parameter Jacobian $J_z$ naturally maintains full row rank. We can state it formally as the following theorem.
\begin{theorem}[Full rank of the empirical NTK]
Let $z_1,\dots,z_N$ be distinct points with $|z_i|\leqslant 1$, and let $f(z;\theta)$ be a feedforward neural network with a real analytic, non-polynomial activation function $\sigma$ and a total of $\mathcal N$ hidden neurons across all layers. If $\mathcal N\geqslant N$, then for a generic choice of parameters $\theta$, the parameter Jacobian $J_z\in\mathbb R^{N\times P}$ has row rank $N$, and the empirical NTK $\Theta_z=J_zJ_z^T/P$ is positive definite: $\lambda_{\min}(\Theta_z)>0$, where $P$ is the total number of parameters, satisfying $P\geqslant\mathcal N$ since every neuron carries at least one parameter.
\end{theorem}

Check Appendix~\ref{sec:thoremproofs} for the proof of this theorem. Physically, the NTK governs the network's expressive bandwidth, i.e., its ability to simultaneously resolve sharp features like the $\rho(770)$ resonance and the smooth, asymptotic spacelike tail without suffering from spectral starvation. Because the conformal mapping prevents the parameter gradients from vanishing in the deep Euclidean limit, the empirical NTK in the conformal space remains strictly positive definite, $\lambda_{\min}(\Theta_z) \geqslant \lambda_0 > 0$, throughout the entire training interval. This ensures that no region of the kinematic phase space becomes a \emph{dead zone} for learning. Under the linearised NTK gradient-flow dynamics, the residual error decays at a rate proportionally bounded by this minimum eigenvalue:
\begin{equation}
\frac{d\mathcal{L}}{dt} \leqslant -\frac{2P\lambda_0}{N}\mathcal{L}.
\end{equation}
This guarantees an exponential decay of the loss: $\mathcal{L}(t) \leqslant \mathcal{L}(0) \exp\left(-\frac{2P\lambda_0}{N} t\right)$. Because no diagonal entry of $\Theta_z$ collapses asymptotically, the trace of the NTK scales with the size of the dataset $N$:
\begin{equation}
    \operatorname{Tr}(\Theta_z) = \sum_{i=1}^N \Theta_{ii} \geqslant N \min_{1 \leqslant i \leqslant N} (\Theta_{ii}) > 0.
\end{equation}

Finally, we note that the neural network can still arbitrarily approximate the true physical form factor $F(s)$ in this new domain.  Because the composition of the form factor and the inverse conformal map, $F(s(z))=F(-4s_{\mathrm{th}}z/(1-z)^2)$, remains complex analytic (holomorphic). The physical target is guaranteed to remain perfectly smooth and free of coordinate-induced singularities. When the infinite physical space is mapped into a closed unit disk, this subset shields the network from the divergent parameter errors from fitting asymptotic tails to infinity. With a smooth target on a finite canvas, the universal approximation theorem guarantees that our neural network can approximate the true pion form factor with arbitrary precision. Crucially, this theorem requires non-polynomial activations; a polynomial network would algebraically collapse and fail to map complex, sharp resonance peaks. The adaptive sine function ($\sigma(x) = \sin(a x)$, where $a$ is a learnable vector) safely bypasses this restriction. In short, the conformal mapping delivers the best of both worlds: the geometric stability required for gradient descent to find a minimum, and the expressivity needed to ensure that minimum accurately reflects the true underlying QCD dynamics.

%%%%%%%%%%%%%%%%%%%%%%%%%%%%%%%%%%%%%%%%%%%%%%%%%%%%%%%
\begin{figure*}
    \centering
    \includegraphics[width=\textwidth]{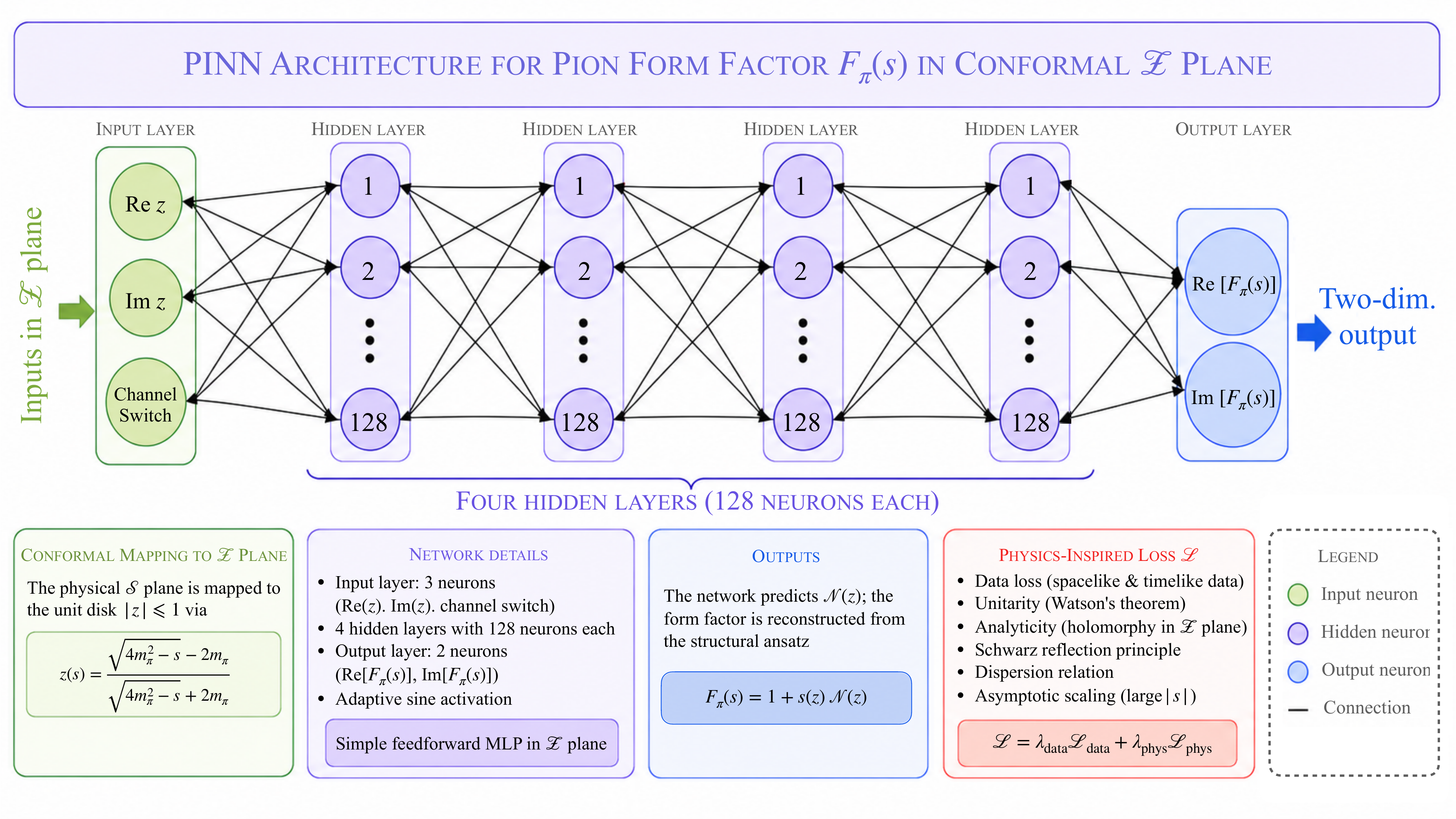}
    \caption{{\bf Network schematic:} PINN architecture for the pion form factor $F_\pi(s)$ in the conformal $\mathcal{Z}$ plane. The physical variable $s$ is first mapped to the unit disk via the conformal bijection $z(s)$. The network takes three inputs -- $\text{Re}\,z$, $\text{Im}\,z$, and a channel switch distinguishing the $e^+e^-$ and $\tau$-decay datasets -- and passes them through four hidden layers of $128$ neurons each with adaptive sine activations. The output $\mathcal{N}(z)$ is combined with the structural ansatz $F_\pi(z) = 1 + s(z)\mathcal{N}(z)$, which hardwires the charge normalisation $F_\pi(0) = 1$. Training minimises a composite loss $\mathcal{L} = \lambda_{\rm data}\mathcal{L}_{\rm data} + \lambda_{\rm phys}\mathcal{L}_{\rm phys}$, where the physics loss enforces unitarity (Watson's theorem), analyticity (Cauchy-Riemann conditions via stochastic collocation), the Schwarz reflection principle, the twice-subtracted dispersion relation, and perturbative QCD asymptotics simultaneously. \label{fig:schematic}}
\end{figure*}
%%%%%%%%%%%%%%%%%%%%%%%%%%%%%%%%%%%%%%%%%%%%%%%%%%%%%%%

%%%%%%%%%%%%%%%%%%%%%%%%%%%%%%%%%%%%%%%%%%%%%%%%%%%%%%%
\subsection{Enforcement of physical constraints}
\noindent
We show the schematic of the simple four-layered, feed-forward network we use in Fig.~\ref{fig:schematic}. We impose the physical and analytical constraints through the loss function. However, we also structurally encode two fundamental constraints directly into the architecture. First, we make a structural ansatz:
\begin{equation}
F_\pi(z)=1+s(z)\,\mathcal{N}(z)
\end{equation}
where $\mathcal{N}(z)$ is the complex-valued output of the neural network and $s(z)$ is the inverse of the conformal map. This formulation guarantees that the charge normalisation, $F_\pi(0)=1$, is satisfied identically, independent of any trainable parameters. Second, the pion form factor must be strictly real on the real axis below the threshold, i.e., $\text{Im}\,F_\pi(s) = 0$ for $s < 4m_\pi^2$, by the Schwarz reflection principle. Because our conformal mapping preserves this complex conjugation, the reality condition holds identically for $z$ on the real axis. We impose this as an exact, architectural hard constraint. For any complex kinematic input $z$, the network evaluates the amplitude using only the absolute value of the imaginary part, $|\text{Im}(z)|$, effectively folding the lower half-plane onto the upper half-plane. To recover the analytic structure, the predicted imaginary part of the form factor is multiplied by $\text{sgn}(\text{Im}(z))$. This forward-pass routing inherently guarantees exact Schwarz reflection globally. 

We construct a loss functional that acts as a regularising boundary in the loss landscape, preventing the network from settling into mathematically degenerate solutions while fitting the experimental data. It is a weighted superposition of experimental data losses and physics-motivated losses:
\begin{equation}
\mathcal{L}_{\mathrm{total}}=\sum_{d\,\in\,{data}}\lambda_d\mathcal L_d+\sum_{p\,\in\,{phys}} \lambda_p \mathcal{L}_{p},
\end{equation}
with $data \equiv$ \{spacelike, timelike\} and $phys \equiv$ \{analyticity, dispersion, moment, positivity, Watson, monotonicity, pQCD, t-asymptotic\}. The weights $\lambda_i$'s are empirically determined and fixed. We define the individual physics constraints below. For brevity, we show each loss term as a squared residual at a single evaluation point; the full loss is its mean over all relevant points for that constraint.

\begin{enumerate}

    \item {\bf Analyticity ($\mathcal L_{\rm analyticity}$):} We enforce the form factor to be smooth and differentiable in the $\mathcal Z$ plane by minimising the violation of the Cauchy-Riemann (CR) conditions:
    \begin{align}
        r\frac{\partial u}{\partial r} = \frac{\partial v}{\partial \theta}, \quad r\frac{\partial v}{\partial r} = -\frac{\partial u}{\partial \theta},   
    \end{align}
    where $z=re^{i\theta}$ with $r\leqslant R_{\max}$ and $\theta\in [-\pi,\pi]$, and $F_\pi(z)=u(z)+iv(z)$, i.e., 
    \begin{align}
        \mathcal{L}_{\text{analyticity}} = \left[ \left( r \frac{\partial u}{\partial r} - \frac{\partial v}{\partial \theta} \right)^2 + \left( r \frac{\partial v}{\partial r} + \frac{\partial u}{\partial \theta} \right)^2 \right].
    \end{align}
    We evaluate all spatial gradients exactly via automatic differentiation, thereby eliminating finite-difference truncation errors. Instead of evaluating the loss on a static, high-density lattice -- which scales poorly ($\mathcal{O}(N^2)$) and allows the neural network to interpolate unphysically or memorise analyticity solely at the grid nodes -- we implement a continuous stochastic collocation method.  At each training epoch, a new coordinate cloud $\{z_j\}_{j=1}^{N_{\rm col}}$ is sampled uniformly from the continuous interior domain $|z| \leqslant R_{\text{max}} = 0.999$, avoiding boundary singularities from the physical branch cut. By dynamically shifting these evaluation points at each step, this Monte Carlo sampling approximates the continuous volume integral over the disk for a large number of epochs ($E\to\infty$). This strategy requires only a small $\mathcal{O}(N_{\rm col})$ point set per iteration (e.g., $10^3$ points over $10^4$ epochs), effectively scanning $\sim 10^7$ unique points. This guarantees sufficient spatial resolution across the domain and prevents localised violations of analyticity.

    \item {\bf Dispersion relation ($\mathcal L_{\rm dispersion}$):} This relation essentially functions as an analytic holographic principle for the pion. It dictates that the entire complex interior is completely determined by its absorptive imaginary part (the physical cross-section) on the boundary branch cut. However, directly applying an unsubtracted dispersion integral exposes the network to severe ultraviolet (UV) instabilities from noisy high-energy extrapolations. To insulate the model, we perform two subtractions at the origin. Because we map the infinite physical branch cut onto the lower conformal boundary arc ($e^{i\theta}$ with $\theta \in [-\pi, 0]$), we can evaluate this causality constraint entirely within the conformal plane. Enforcing charge normalisation ($u(0)=1$), the dispersive reconstruction for a spacelike point $s$ is given by:
    \begin{align*}
        u_{\rm disp}(s) =&\ 1 + \frac{1}{6}\langle r_\pi^2 \rangle s + \frac{s^2}{\pi} \int_{4 m_\pi^2}^\infty \frac{v(s')}{(s')^2 (s' - s)} ds',\label{eq:dispersion}
    \end{align*}
    which, in the $\mathcal Z$ plane, can be written as,
     \begin{align}
     u_{\rm disp}(z)=&\ 1 + s(z)\left.\frac{dF_\pi(s)}{ds}\right|_{s=0} + \Delta_{\rm disp}(z).   
    \end{align}
    Rather than fixing the slope with an empirical parameter, we extract $F_\pi'(0)$ dynamically from the network via automatic differentiation to guarantee self-consistency and evaluate $\Delta_{\text{disp}}(z)$ as,
    \begin{align*}
        \Delta_{\text{disp}}(z) = \frac{s^2(z)}{2\pi s_{\text{th}}} \int_{-\pi}^0 \frac{v(e^{i\theta'}) |\sin\theta'|(1-\cos\theta')}{2s_{\text{th}} - s(z)(1-\cos\theta')} d\theta'.   
    \end{align*}
    We thus minimise the following squared loss:
    \begin{align}
        \mathcal{L}_{\text{dispersion}} = \left[ u(z) - u_{\rm disp}(z) \right]^2.
    \end{align}
    The suppression in the integration kernel heavily damps high-energy contributions, stabilising the loss landscape and focusing the network's expressive power on the resonant physics. 

    \item 
    {\bf Higher-moment sum-rules ($\mathcal L_{\rm moment}$):} 
    Neural networks can suffer from spectral bias, a tendency for gradient descent to preferentially learn smooth, low-frequency functions while severely underestimating sharp, high-frequency features. In our context, this bias can cause the network to artificially flatten the local curvature at the origin, introducing unphysical parameter degeneracies (such as allowing the model to compensate incorrectly for the $\rho-\omega$ mixing phase). To break this degeneracy, we explicitly constrain the network's dynamically generated local derivatives to match the global dispersive sum rules. 
    
    Along the physical real axis, the analyticity of the pion form factor lets us evaluate the generalised $n$-th spectral moment $I_n$ ($n\geqslant0$) as:
    \begin{align}
        I_n = \left.\frac{1}{n!}\frac{d^n}{ds^n}F_\pi(s)\right|_{s=0}=\frac{1}{\pi} \int_{s_{\rm th}}^{\infty} \frac{v(s')}{(s')^{n+1}} ds',
    \end{align}
    with $I_0 = F_\pi(0) = 1$. This provides tight topological constraints to control how rapidly the spectral density flattens out as it transitions into the perturbative QCD continuum. As with the dispersion relation, integrating to infinity numerically is unstable. We therefore define $\mathcal L_{\rm moment}$ by compactifying the integral after projecting it onto the conformal boundary $z' = e^{i\theta'}$, i.e.,
    \begin{align}
        \mathcal{L}_{\rm moment} = \sum_{n=1}^4 & \Bigg[ \frac{1}{\pi} \int_{-\pi}^{0} \frac{v(e^{i\theta'})s(\theta') |\sin\theta'|}{s^{n+1}(\theta')(1 - \cos\theta')}  \, d\theta'- \frac{F_\pi^{(n)}(0)}{n!}\Bigg]^2,
     \end{align}
     where $F^{(n)}_\pi$ denotes the $n$th derivative of $F_\pi$ with respect to $s(z)$. These integral bounds restrict the network from introducing unphysical, wide spectral artefacts in deep inelastic regions where experimental data is sparse.

     \item {\bf Spectral positivity ($\mathcal L_{\rm positivity}$):} Unitarity dictates that since the absorptive part of the form factor on the branch cut corresponds to physical cross-sections, it must be strictly non-negative:
    \begin{align}
        v(s) \geqslant 0, \quad \text{for } s > 4m_\pi^2.
    \end{align}
    which translates to the lower unit-semicircle boundary in the $\mathcal Z$ plane. Since an inequality boundary is difficult to enforce dynamically, we implement it as a soft topological penalty in the loss function using a rectified linear unit (ReLU) functional (one-sided hinge loss):
    \begin{align}
        \mathcal{L}_{\text{positivity}} = \left[\max\left(0, -v(z)\right) \right]^2,
    \end{align}
    The penalty is zero when the physics is respected, but grows quadratically when the network strays into unphysical values.

    \item 
    {\bf Watson’s theorem ($\mathcal L_{\rm Watson}$):} In the region where elastic scattering dominates ($4m_\pi^2 < s \lesssim 1 \text{ GeV}^2$), Watson's theorem sets $\arg F_\pi(z) = \delta_1^1(s(z))$ [see Eq.~\eqref{eq:watson_relation}]. However, computing $\tan^{-1}(v(z)/u(z))$ can introduce coordinate singularities for small amplitudes and artificial discontinuous jumps across branch cuts. To bypass these, we reformulate the constraint as a geometric projection in the complex plane, where we minimise the perpendicular projection of the form factor relative to the target phase angle as,
    \begin{align}
        \mathcal{R}_{\phi} = u(z_k)\sin\delta_1^1(s_k) - v(z_k)\cos\delta_1^1(s_k)=0.   
    \end{align}
    This scales smoothly as $|F_\pi|\sin(\delta_1^1 - \phi)$,
    eliminating any singular denominators. Because $\mathcal{R}_{\phi}(s) = 0$ is satisfied for both $\delta_1^1$ and $\delta_1^1 + \pi$, we enforce a directional constraint to avoid the wrong branch. We require the parallel projection to be strictly positive:
    \begin{align}
        \mathcal P_\phi= u(z)\cos\delta_1^1(s(z)) + v(z)\sin\delta_1^1(s(z)), 
    \end{align}
    We apply the following loss functional:
    \begin{align}\label{eq:watson-loss}
        \mathcal L_{\rm Watson} = \mathcal R^2_{\phi}+\left[\max(0,-\mathcal P_\phi)\right]^2.
    \end{align}
    By optimising these linear combinations of the form factor components ($u, v$), the PINN stably tracks the rapid phase variations across the $\rho(770)$ resonance.

    \item {\bf Phase monotonicity ($\mathcal{L}_{\rm monotonicity}$):} To prevent unphysical phase oscillations or numerical back-tracking in regions where experimental phase-shift data is sparse or plagued by inelastic channel openings, we implement a phase monotonicity loss along the lower unit-semicircle boundary in the conformal space as,
    \begin{align}
        \mathcal{L}_{\rm monotonicity} =&\ \left[ \max\left(0, -\frac{d}{d\theta}\Big(\arg F_\pi(e^{i\theta})\Big)\right) \right]^2\nonumber\\
        =&\ \left[\max\left(0,v\frac{du}{d\theta}-u\frac{dv}{d\theta}\right)\right]^2,
    \end{align}
    for $\theta\in[-\pi,0]$. It forces the angular derivative of the form factor's phase to be non-negative.

    \item {\bf Asymptotic pQCD ($\mathcal{L}_{\rm PQCD}$):} 
    We also consider the asymptotic pQCD scaling behaviour in the deep spacelike region. From Eq.~\eqref{eq:pqcd_ff} (and Appendix~\ref{sec:nlo}), we observe that, 
    \begin{align*}
        \frac{d}{dQ^2} \left[ \frac{Q^2\,u(Q^2)}{\alpha_s\left(Q^2/21\right)\left[1+0.18\,\alpha_s\left(Q^2/21\right)\right]} \right]\approx 0, 
    \end{align*}
    in the domain of pQCD. Thus, we minimise the following loss functional in the conformal space for $z\geqslant0.82$ (or, $Q^2\gtrsim 8$ GeV$^2$), 
    \begin{align}
        \mathcal{L}_{\rm pQCD} =&\ \Bigg(\frac{(1-z)^3}{(1+z)} \frac{d}{dz} \left[ \frac{z}{(1-z)^2}\frac{u(z)}{\alpha_s\left(Q^2(z)/21\right)}\right.\nonumber\\
        &\ \left.\times\frac{1}{\left[1+0.18\,\alpha_s\left(Q^2(z)/21\right)\right]} \right] \Bigg)^2, 
    \end{align}
    This way, we don't bias the network by the pQCD normalisation, only the scaling.

    \item {\bf Timelike asymptotic ($\mathcal{L}_{\rm{t-asymptotic}}$):}  Power counting rules of pQCD imply that the real part of the form factor drops as $1/s$ as $s\to \infty$, i.e., $u\sim 1/s$ deep along the edge of the branch cut. With $ s(\theta')=8m_\pi^2/(1-\cos\theta')$ on the $z'=e^{i\theta'}$ circle in the conformal plane, we have  $\theta'^2\approx16m_\pi^2/s$ for small $\theta'$. Hence, we demand that, $u$ grows as $\theta'^2$ near $\theta'=0$ (i.e, $s\to\infty)$. In other words, we require $u(1)=u'(1)=0$ so that we have 
    \begin{align}
        u(z') \approx -\frac12u''(1)(z'-1)^2 \approx -\frac12u''(1)\theta'^2 
    \end{align}
    near $z'=1$. Thus, we use the following loss function
    \begin{align}
        \mathcal{L}_{\rm{t-asymptotic}} = \left[u^2+\left(\frac{du}{dz'}\right)^2\right]_{z'=1}.
    \end{align}
\end{enumerate}
%%%%%%%%%%%%%%%%%%%%%%%%%%%%%%%%%%%%%%%%%%%%%%%%%%%%%%%
\begin{table*}[t]
\centering
\caption{\label{tab:spacelike-data} Experimental datasets used for training the neural network.}
{\renewcommand{\arraystretch}{1.5}
\begin{tabular*}{\textwidth}{@{\extracolsep{\fill}}l  lrr}
\hline\hline
Scattering channel&Experiment& Data points& Energy range (in GeV$^2$)\\\hline
{\bf Spacelike data}&&&$Q^2=-s$\\
\hline
\multirow{2}{*}{$e\pi\to e\pi$}
    & NA7~\cite{NA7:1986vav} & $45$ & $0.015-0.253$\\
    & Fermilab~\cite{Dally:1981ur} &20 &$0.03-0.07$\\ 
\hline
\multirow{4}{*}{$e^-p \to e^- \pi^+n$}
    &CEA~\cite{Brown:1973wr} &5 &$0.176-1.188$ \\
    & Cornell '71 data~\cite{Bebek:1974iz}&6 & $0.62-2.015$ \\
    &JLab~\cite{JeffersonLab:2008gyl}  & 8& $0.60-2.45$ \\ 
    & WSL at Cornell University~\cite{Bebek:1974ww} &6 & $1.2-4.0$\\
\hline
{\bf Timelike data}&&&$q^2=s$\\
\hline
\multirow{2}{*}{$\tau^- \to \pi^-\pi^0\nu_\tau$}
    &Belle Collaboration~\cite{Belle:2008xpe} &$62$ &$0.088-3.125$\\
    &CLEO Collaboration~\cite{CLEO:1999dln}&$41$&$0.11-2.81$\\
\hline
\multirow{7}{*}{$e^+e^- \to \pi^+\pi^-$}
    &CMD-2 Collaboration~\cite{CMD-2:2006gxt} &$29$ &$0.36-0.94$\\
    &CMD-3 Collaboration~\cite{CMD-3:2023rfe}&$209$ &$0.11-1.44$\\
    &DM2 detector~\cite{DM2:1988xqd}&$17$ &$1.82 - 4.52$\\
    &CLEO-c~\cite{Seth:2012nn}&$2$&$14.2,17.4$\\ 
    &ADONE storage ring at Frascati~\cite{Bollini:1975pn}&$12 $&$1.44-9$\\
    &VEPP-2M, OLYA- 85~\cite{Barkov:1985ac}&$79$ &$0.16-1.95$ \\ 
    &VEPP-2M, OLYA- 78~\cite{Bukin:1978sq}&$29$ &$0.61-1.77$ \\
    &SND detector~\cite{SND:2020nwa} &$36$ &$0.28-0.78$ \\
\hline
\multirow{2}{*}{$e^+ e^- \to \pi^+\pi^-\gamma$}
    &KLOE detector~\cite{KLOE:2010qei}&$75$ &$0.1 - 0.85$\\
    &BaBar~\cite{BaBar:2009wpw,BaBar:2012bdw}&$337$ &$0.09-9$\\ 
    &BESIII detector~\cite{BESIII:2015equ} &$60$ &$0.36-0.81$ \\
\hline\hline
\end{tabular*}}
\end{table*}
%%%%%%%%%%%%%%%%%%%%%%%%%%%%%%%%%%%%%%%%%%%%%%%%%%%%%%%

%%%%%%%%%%%%%%%%%%%%%%%%%%%%%%%%%%%%%%%%%%%%%%%%%%%%%%%
 \begin{figure*}
    \captionsetup[subfigure]{labelformat=empty}
    \centering
    \subfloat[{\bf (a)} $|F_\pi(s)|$]{\includegraphics[width=\textwidth]{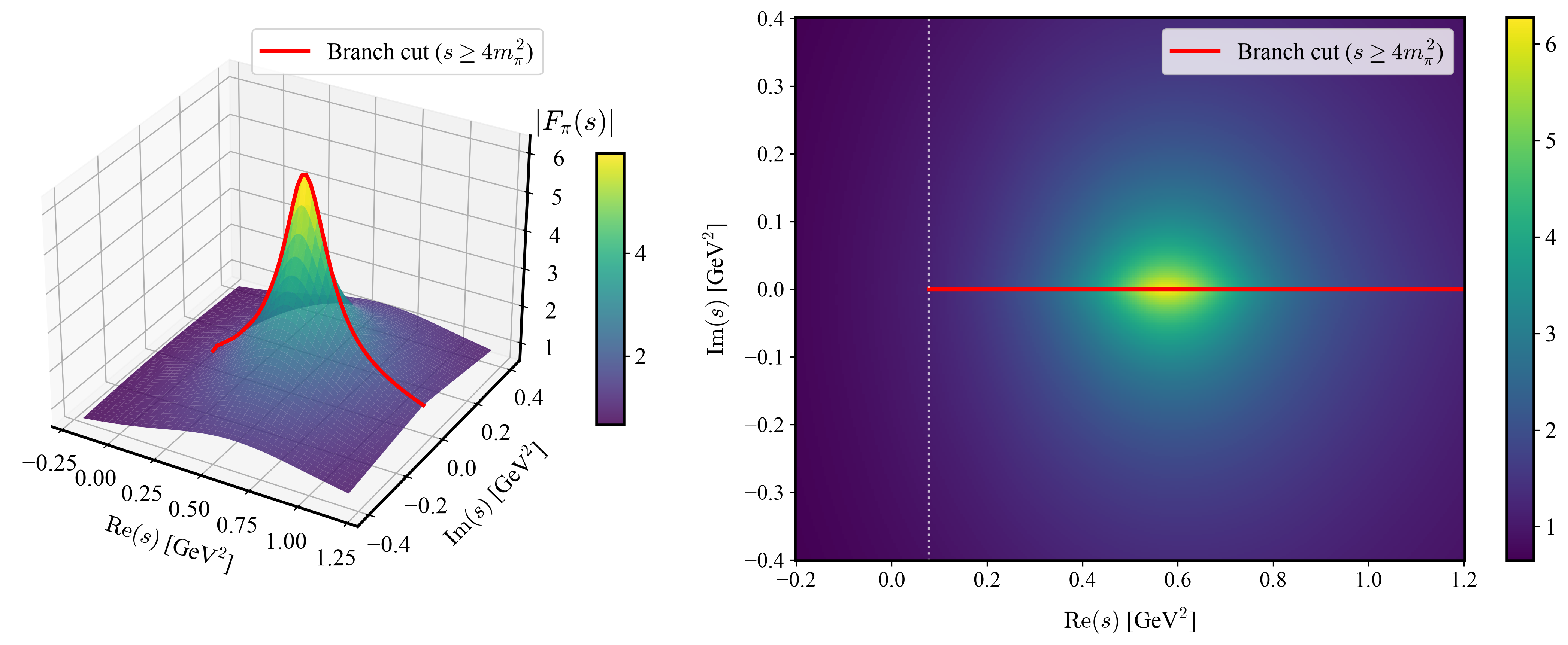}\label{fig:complex_plane_3d}}\\
    \subfloat[{\bf (b)} $\text{Arg}\left(F_\pi(s)\right)$]{\includegraphics[width=\textwidth]{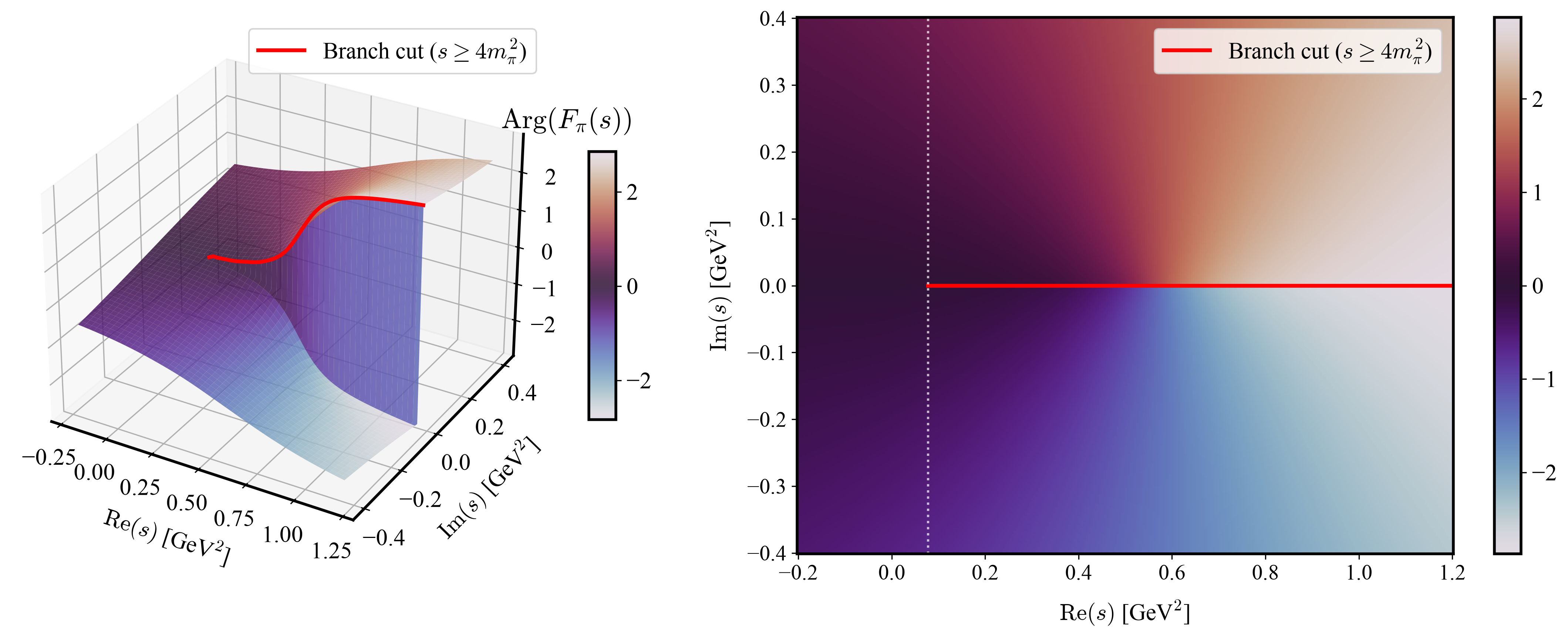}\label{fig:complex_plane_phase_3d}}
    \caption{{\bf (a)} The pion form factor analytically continued on the complex $\mathcal S$ plane; {\bf (b)} its argument validating the Schwarz relation. The peak is at the $\rho(770)$ resonance. The discontinuity across the branch cut is clearly visible in the argument. \label{fig:3d}}
\end{figure*}
%%%%%%%%%%%%%%%%%%%%%%%%%%%%%%%%%%%%%%%%%%%%%%%%%%%%%%%

%%%%%%%%%%%%%%%%%%%%%%%%%%%%%%%%%%%%%%%%%%%%%%%%%%%%%%%
\section{Experimental data}\label{sec:expt-data}
\noindent
We train our network on both spacelike (primarily from electroproduction data) and timelike (including high-precision cross-section results from BaBar, CMD-3, and BESIII collaborations) datasets (see Table~\ref{tab:spacelike-data} for the references). Spacelike data are normally expressed in terms of $Q^2=-q^2$; $F_\pi(Q^2)$ is well-measured up to $Q^2$ values as low as $0.28$ GeV$^2$ by elastic electron-pion scattering. At large  $Q^2$, $F_\pi(Q^2)$ is determined through pion electroproduction from a nucleon target, $e^-p \to e^-\pi^+n$. The longitudinal part of the pion electroproduction cross section, $\sigma_L$, contains the pion exchange process, where a virtual photon couples to a virtual pion within the nucleon, from which these experiments extracted $F_\pi$. 

For the timelike region, we use $e^+ e^-$ scattering data: the direct energy-scan measurements from SND, CMD-2, and CMD-3, and the initial-state radiation (ISR)-based  \emph{radiative return} measurements at fixed baseline energies from BaBar, KLOE, and BESIII. For the $e^+e^-$ datasets, we use either the bare form factor, undressed of vacuum polarisation effects, or the one obtained after correcting for final-state radiation (FSR) effects, as applicable, by the following formula
 \begin{align}
 |F_\pi(s)|^2 = |F_\pi^{\text{exp}}|^2 |1-\Pi(s)|^2\left(\frac{\pi}{\pi+\alpha_{em}\mathfrak f(s)}\right),\label{eq:fsr}
 \end{align}
where 
\newcommand{\dilog}{\mathrm{Li}_2}
\begin{align*}
\mathfrak f(s) =&\ \frac{3(1+\sigma_\pi^2(s))}{2\sigma_\pi^2(s)} - 4 \log\sigma_\pi(s) + 6 \log \frac{1+\sigma_\pi(s)}{2} \nonumber \\
&\ + \frac{1+\sigma_\pi^2(s)}{\sigma_\pi(s)} F(\sigma_\pi(s))  - \frac{(1-\sigma_\pi(s))}{4\sigma_\pi^3(s)} \nonumber \\ 
&\ \times\big(3+3\sigma_\pi(s)-7\sigma_\pi^2(s)+5\sigma_\pi^3(s)\big) \log\frac{1+\sigma_\pi(s)}{1-\sigma_\pi(s)}, 
\end{align*}
\begin{align*}
    \sigma_\pi(s) =&\ \sqrt{1-\frac{4 m_\pi^2}{s}},\nonumber \\
F(x) =&\ -4\dilog(x)+4\dilog(-x)+2\log x \log\frac{1+x}{1-x} \nonumber\\
&\ + 3 \dilog\Big(\frac{1+x}{2} \Big) - 3\dilog\Big(\frac{1-x}{2} \Big) + \frac{\pi^2}{2},  \nonumber \\
\dilog(x) =&\ - \int_0^x dt \frac{\log(1-t)}{t}. \nonumber
\end{align*}
The term $|1-\Pi(s)|^2$ with the polarisation operator $\Pi(s)$ excludes the effect of leptonic and hadronic vacuum polarisation~\cite{Jegerlehner:2008rs}, so that one obtains the bare cross section.

It is well-known that there is a systematic tension between the high-precision $e^+e^- \to \pi^+\pi^-$ cross-section measurements from BaBar and CMD-3, with CMD-3 yielding a significantly higher form factor magnitude ($\vert{}F_\pi(s)\vert{}$) across the $\rho$ resonance. BaBar uses ISR at a fixed centre-of-mass energy, and CMD-3 employs a step-by-step energy scan. Older direct-scan measurements like CMD-2 agree well with BaBar and $\tau$-decay data, indicating the shift is specific to CMD-3 rather than an inherent flaw in the energy-scan methodology. This persistent discrepancy hints towards some unaccounted systematic errors in at least one framework, likely originating from sub-leading radiative photon corrections or integrated beam luminosity calibrations.

To evaluate the sensitivity of the learned representation to dataset-specific tensions, we train the network on five baseline sets. Two on the $e^+e^-$ data: 
\begin{enumerate}
    \item 
    Set A (spacelike + $e^+e^-$ data excluding CMD-3) and
    
    \item 
    Set B (spacelike + only CMD-3 data).

\end{enumerate}

The form factor that appears in the $\tau^-\to \pi^-\pi^0\nu_\tau$ decay is different from the electromagnetic form factor $F_\pi(s)$ in the $e^+e^-$ annihilation. The $\tau$ decay probes the magnitude of the weak form factor $f_\pi^-(s)$. These two functions are related as~\cite{Cirigliano:2002pv}:
\begin{align}
    |F_\pi(s)|^2 = |f_\pi^-(s)|^2 \times R_{\text{IB}}(s),\label{eq:ribmixing}
\end{align}
where $R_{\text{IB}}(s)$ compensates for the isospin-breaking:
\begin{align*}
R_{\text{IB}}(s) = \frac{1}{G_{\text{EM}}(s)} \frac{\beta_{\pi^+\pi^-}^3(s)}{\beta_{\pi^-\pi^0}^3(s)}  \left|\frac{F_V(s)}{f_+(s)}\right|^2.    
\end{align*}
Here, $\beta_{\pi\pi}^3(s)$ accounts for the phase-space difference arising from the $\pi^\pm-\pi^0$ mass splitting, $G_{\text{EM}}(s)$ isolates long-distance electromagnetic radiative corrections, and the form factor ratio $\left\vert{}F_V(s)/f_+(s)\right\vert{}^2$ captures intrinsic hadronic isospin-breaking effects, including $\rho^\pm-\rho^0$ mass and width differences ($\Delta m_\rho, \Delta \Gamma_\rho$) as well as $\rho-\omega$ mixing present in the neutral channel. 

\begin{enumerate}
  \setcounter{enumi}{2}
  \item[3-4.] 
  Set C (spacelike + $\tau$ decay data): We train in two ways: 
  \begin{itemize}
      \item[--] 
      \emph{Without channel switch} (Set C$_{R_{\rm IB}}$): we correct the $\tau$ data for isospin-violation explicitly by $R_{\rm IB}$ and use that for training, and 
      \item[--]
      \emph{With channel switch} (Set C$_{\rm raw}$): we use the architectural flexibility of the PINN  by making the neural network directly approximate the pure isovector form factor, $F_\pi^{I=1}(s)$, from the raw $\tau$ data. We handle the dataset differences by using the mixing parameter $\alpha_{\rho\text{-}\omega}$ in Eq.~\eqref{eq:mixing} as a conditional architectural switch in the forward pass. This way, the network handles the two data streams seamlessly, creating a joint representation, with $\alpha_{\rho\text{-}\omega} = 0$ for the $\tau$-decay channel and $\alpha_{\rho\text{-}\omega} \neq 0$ for the $e^+e^-$ annihilation channel (see Fig.~\ref{fig:schematic}).
  \end{itemize}
\end{enumerate}
We also train on all datasets combined.
\begin{enumerate}
\setcounter{enumi}{4}
  \item 
  Set D (spacelike + $e^+e^-$ data including CMD-3 + raw $\tau$ decay data with channel switch): We combine all available datasets to evaluate the overall model consistency. Similar to Set C$_{\rm raw}$, the $\tau$ decay data is incorporated without explicit $R_{\rm IB}$ corrections, relying on the conditional mixing parameter $\alpha_{\rho\text{-}\omega}$ to handle the channel transition directly within the network architecture.
\end{enumerate}

%%%%%%%%%%%%%%%%%%%%%%%%%%%%%%%%%%%%%%%%%%%%%%%%%%%%%%%
\section{Results}\label{sec:results}
\subsection{Global analytic structure and the zero-free landscape}
\noindent
We first recover the global analytic structure of the extracted pion electromagnetic form factor, $F_{\pi}(s)$, across the complex $s$ plane. Fig.~\ref{fig:complex_plane_3d} shows its absolute magnitude across the first Riemann sheet. The multiparticle branch cut lies along the positive real axis, starting from the two-pion threshold across the peak on the cut near $\text{Re}(s) \approx 0.6$ GeV$^2$, reproducing the profile of the $\rho(770)$ vector resonance.\footnote{Because the pole sits beneath the branch cut on the unphysical second Riemann sheet, the observed profile on the first sheet remains smooth and finite.} Constrained by $\mathcal L_{\rm analyticity}$, the PINN yields a stable reconstruction without introducing unphysical poles or numerical artefacts away from the real axis.

The phase topology mapped in Fig.~\ref{fig:complex_plane_phase_3d} confirms that the network output adheres to the Schwarz reflection principle. Approaching the cut from the upper half-plane yields a positive phase, whereas approaching from the lower half-plane produces a mirror-image negative value. Below the elastic threshold, the phase vanishes identically on the real axis. Beyond the threshold, the branch cut opens, splitting the phase and revealing a sharp, step-like cliff across the real axis. This is enforced by design, as we only calculate the form factor in the upper half of the complex plane; the form factor in the lower half is the complex conjugate of the upper half. This behaviour is explained by Watson's final-state interaction theorem. In the low-energy elastic region, the phase is locked to the scattering phase shift of two interacting pions, rising smoothly with energy, passing through $\pi/2$ around the $\rho(770)$ resonance and continuing toward $\pi$. The phase monotonicity loss ($\mathcal{L}_{\text{monotonicity}}$) ensures this stable angular trajectory in both elastic and inelastic domains.

%%%%%%%%%%%%%%%%%%%%%%%%%%%%%%%%%%%%%%%%%%%%%%%%%%%%%%%
\begin{figure}[!htb]
    \centering
    \includegraphics[width=\columnwidth]{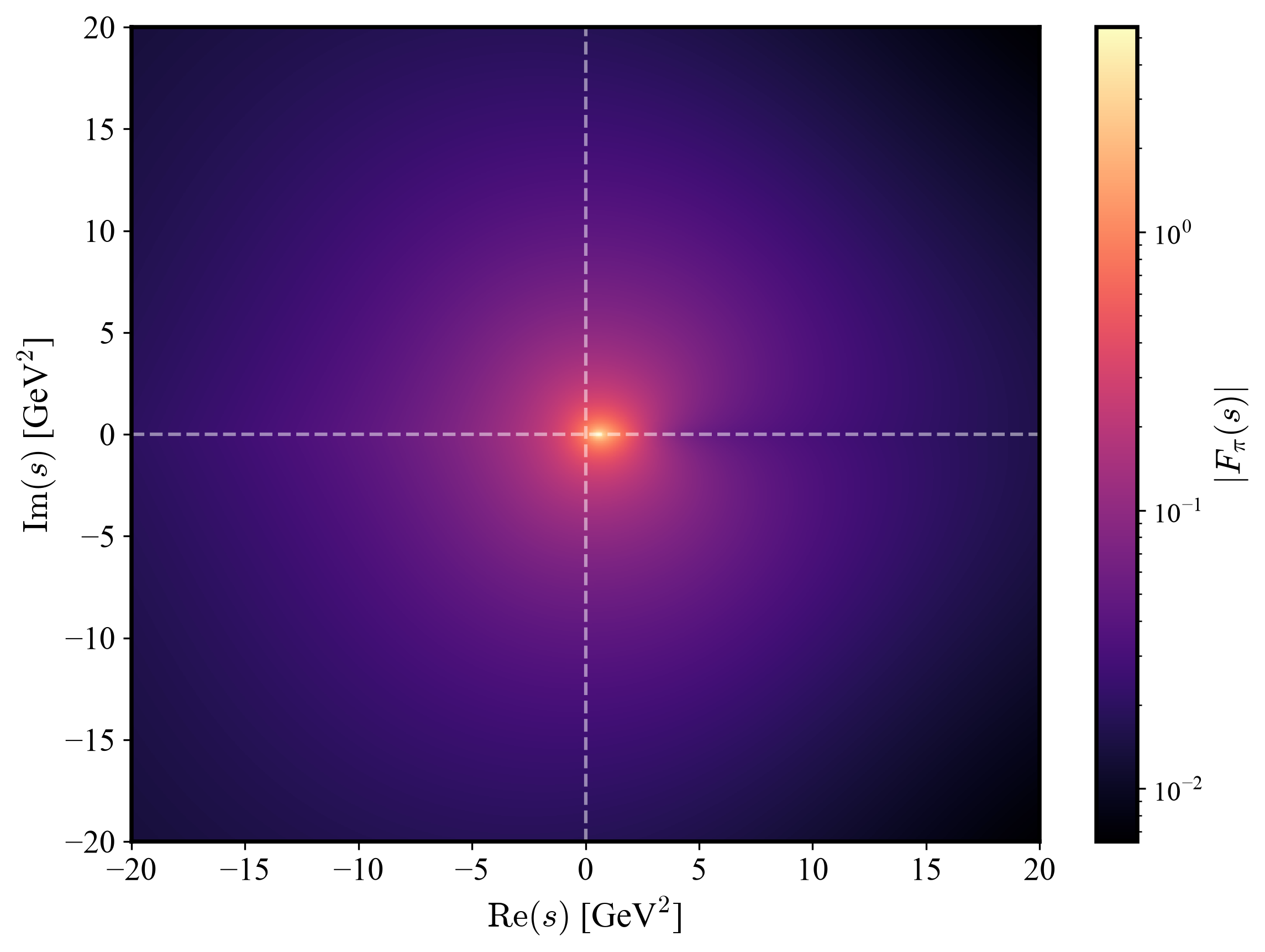}
    \caption{\label{fig:no_zeros} Heat map of $|F_\pi(s)|$ confirming that the predicted form factor contains no zeros in the interior of the $\mathcal S$ plane.}
%\end{figure}
%%%%%%%%%%%%%%%%%%%%%%%%%%%%%%%%%%%%%%%%%%%%%%%%%%%%%%%
\vspace{0.75cm}
%%%%%%%%%%%%%%%%%%%%%%%%%%%%%%%%%%%%%%%%%%%%%%%%%%%%%%%
%\begin{figure}[!t]
    \centering
    \includegraphics[width=\columnwidth]{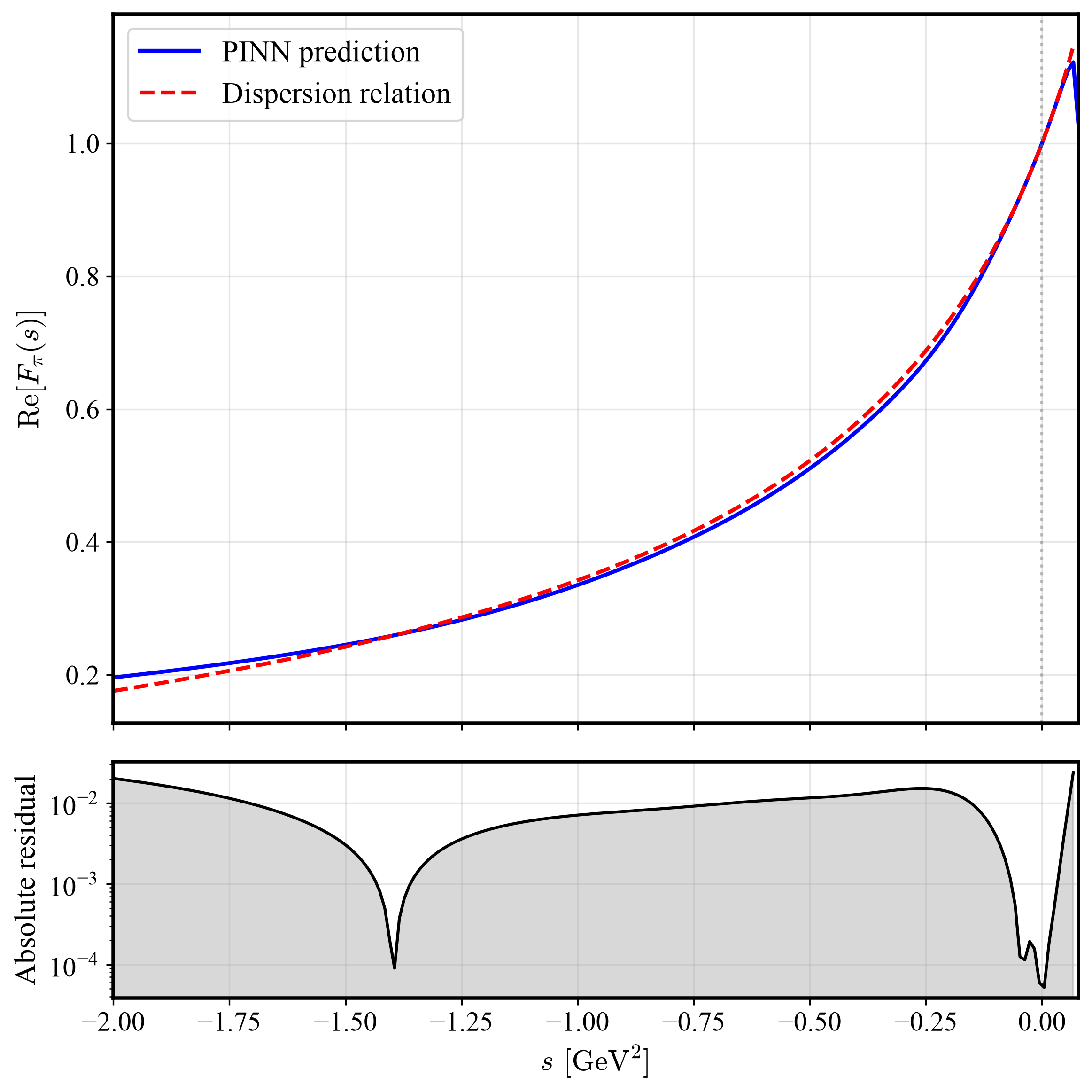}
    \caption{\label{fig:dispersion} Convergence of the dispersion relation loss.}
\end{figure}
%%%%%%%%%%%%%%%%%%%%%%%%%%%%%%%%%%%%%%%%%%%%%%%%%%%%%%%
%%%%%%%%%%%%%%%%%%%%%%%%%%%%%%%%%%%%%%%%%%%%%%%%%%%%%%%

Fig.~\ref{fig:no_zeros} illustrates the logarithmic magnitude landscape over a wide kinematic range. Since the network has not been trained on any zero-free constraint, the absence of unphysical destructive interference patterns provides an independent, data-driven support for the zero-free hypothesis. We see a continuous drop in the deep-spacelike and complex-UV regions, aligning with pQCD scaling expectations ($1/s$).

We also test the network's prediction against the dispersion relation [Eq.~\eqref{eq:dispersion}]. Fig.~\ref{fig:dispersion} illustrates the real component, $\text{Re}[F_\pi(s)]$ up to the two-pion threshold, $s \in [-2.0, 4m_\pi^2]$. We see that the absolute residual, $|\text{Re}[F_\pi(s)]_{\text{PINN}} - \text{Re}[F_\pi(s)]_{\text{disp.}}|$, remains tightly bounded. This adds further evidence that the network has constructed a holistic amplitude on the first Riemann sheet where the local real and imaginary predictions are holomorphic.

% %%%%%%%%%%%%%%%%%%%%%%%%%%%%%%%%%%%%%%%%%%%%%%%%%%%%%%%
\subsection{Form factor across spacelike and timelike regions}
\noindent
Fig.~\ref{fig:global_st} illustrates the predicted pion form factor $|F_\pi(s)|^2$ across the spacelike and timelike regions alongside the experimental datasets. Using the $e^+e^-$ annihilation data for the timelike region (Set A), the network produces a curve that satisfies the charge normalisation and smoothly passes through the low-energy spacelike points, reproducing the $\rho(770)$ vector resonance. (Unless specifically noted otherwise, all results presented in Sec.~\ref{sec:results} are computed using dataset Set A.) The smoothness stems from the dispersion relation that acts as an analytical anchor; the shape of the curve in the timelike region dictates the spacelike behaviour, forcing the network to maintain a holomorphic and physically plausible transition between the two regions.
% %%%%%%%%%%%%%%%%%%%%%%%%%%%%%%%%%%%%%%%%%%%%%%%%%%%%%%%
\begin{figure}[!t]
    \centering
    \includegraphics[width=\columnwidth]{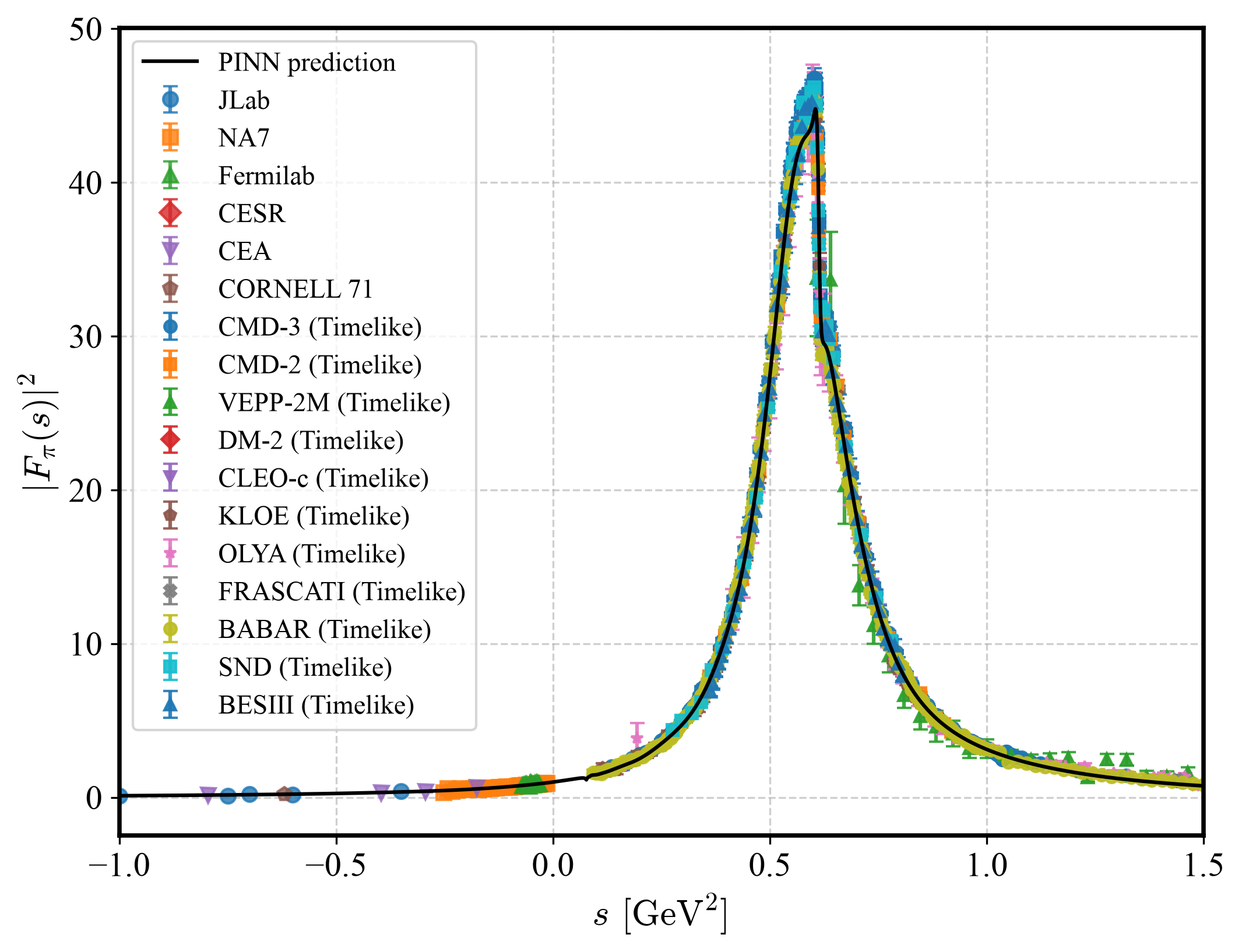}
    \caption{\label{fig:global_st}  The predicted pion electromagnetic form factor squared $|F_\pi(s)|^2$ (black curve) with its associated $1\sigma$ uncertainty band (imperceptible at this scale as $\sigma \sim 10^{-2}$), plotted alongside experimental data across both spacelike ($s < 0$) and timelike ($s > 4m_\pi^2$) kinematic regions. The plot highlights the continuous transition through the low-energy region, the primary $\rho(770)$ resonance peak, and the narrow $\rho-\omega$ interference structure around $s \approx 0.6\text{ GeV}^2$.}
\end{figure}
% %%%%%%%%%%%%%%%%%%%%%%%%%%%%%%%%%%%%%%%%%%%%%%%%%%%%%%%
% %%%%%%%%%%%%%%%%%%%%%%%%%%%%%%%%%%%%%%%%%%%%%%%%%%%%%%%
\begin{figure}[!]
    \centering
    \includegraphics[width=\columnwidth]{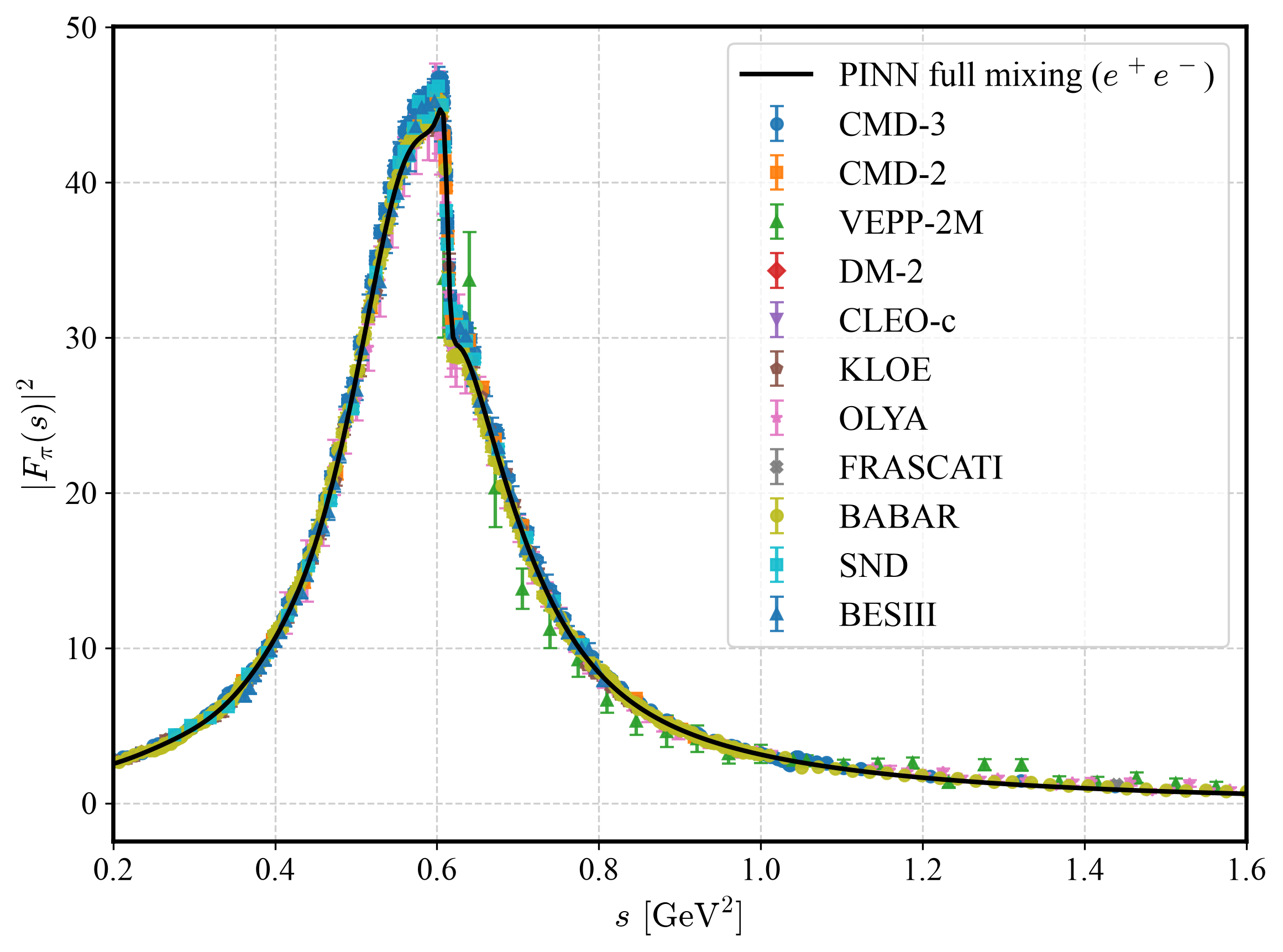}
    \caption{PINN prediction for the timelike region with the $e^+e^-$ datasets showing the $\rho$ resonance along with the \emph{dip-bump} pattern from the $\rho$--$\omega$ interference. \label{fig:timelike}}
\end{figure}
%%%%%%%%%%%%%%%%%%%%%%%%%%%%%%%%%%%%%%%%%%%%%%%%%%%%%%%

The network natively isolates the strong-interaction physics (the isovector $I=1$ state) and, at the same time,  captures the interference pattern from $\rho\text{--}\omega$ mixing. As shown in Fig.~\ref{fig:global_st}, adjusting the $\alpha_{\rho-\omega}$ switch from Eq.~\eqref{eq:mixing} allows the architecture to seamlessly incorporate the isoscalar ($I=0$) $\rho\text{--}\omega$ mixing mechanism without requiring any changes to the underlying neural network core. This mixing leaves a highly localised \emph{dip-bump} ripple directly on the steep slope of the dominant $\rho$ peak.

Fig.~\ref{fig:timelike} shows this explicitly for the electron-positron annihilation data [$e^+e^- \rightarrow \pi^+\pi^-(\gamma)$]. With the switch active ($\alpha_{\rho-\omega} \neq 0$), the network tracks the abrupt phase swing and the small (less than a per cent) interference ripple with high precision, demonstrating that it has learned a universal, underlying representation of the form factor rather than memorising individual datasets. Setting $\alpha_{\rho-\omega} = 0$ instead makes the network correctly ignore all interference, tracking the pure vector resonance smoothly over the data points; this is explored further in Sec.~\ref{sec:ablation}.

%%%%%%%%%%%%%%%%%%%%%%%%%%%%%%%%%%%%%%%%%%%%%%%%%%%%%%%
\begin{figure}[!t]
    \centering
    \includegraphics[width=\columnwidth]{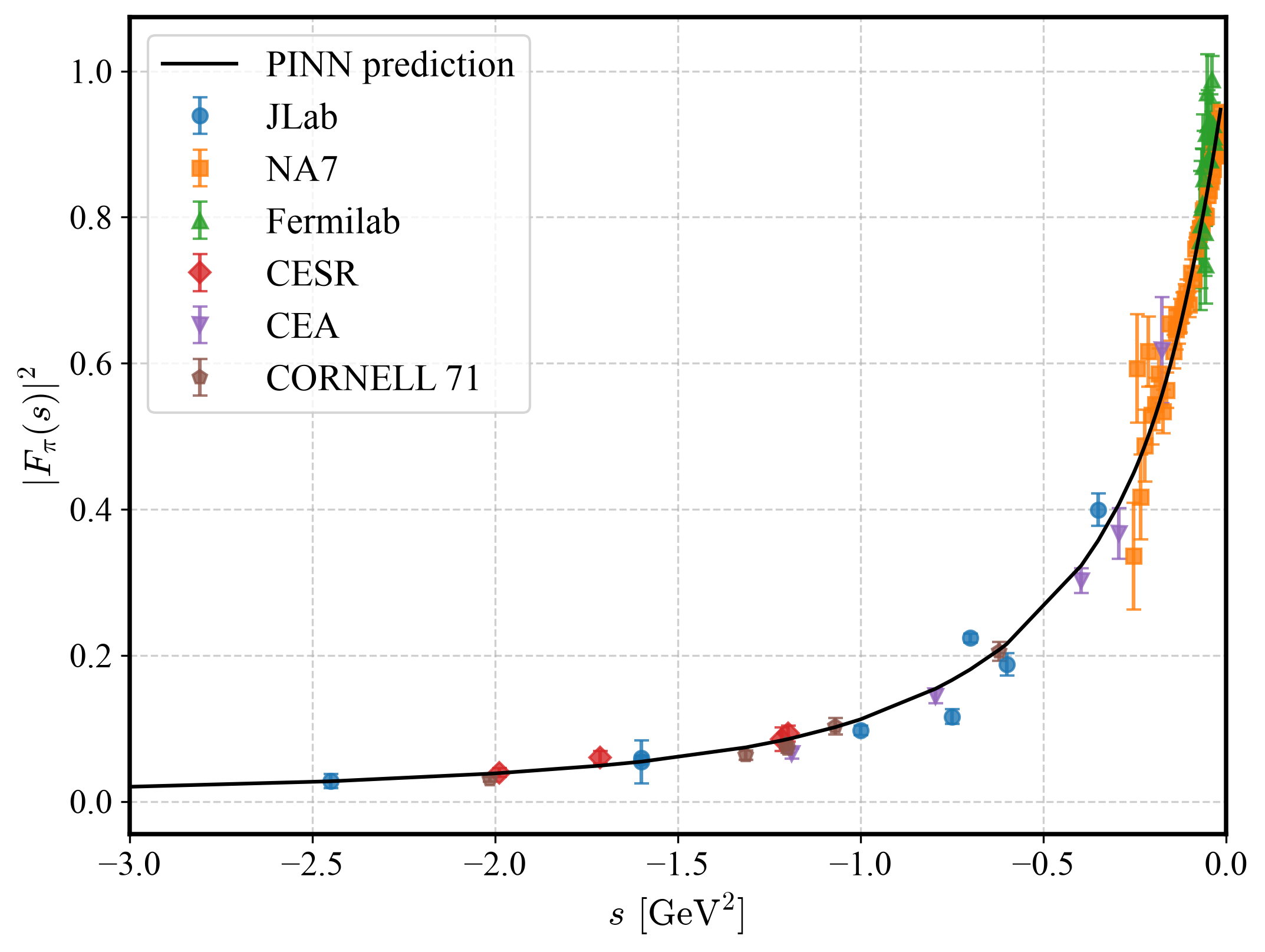}
    \caption{PINN prediction for the spacelike region, overlaid on the experimental data.\label{fig:spacelike}}
\end{figure}
To complete the physical bridge, Fig.~\ref{fig:spacelike} illustrates the spacelike region of $F_\pi$ alongside the experimental electroproduction data. The network produces a stable, monotonic damping that effectively matches the experimental points. The absence of unphysical oscillations or sudden upturns is consistent with the short-distance limits of quantum field theory
% %%%%%%%%%%%%%%%%%%%%%%%%%%%%%%%%%%%%%%%%%%%%%%%%%%%%%%%
% %%%%%%%%%%%%%%%%%%%%%%%%%%%%%%%%%%%%%%%%%%%%%%%%%%%%%%%
\subsection{The spacelike asymptotics}
\noindent
We analyse the asymptotic scaling of the charged pion electromagnetic form factor in Fig.~\ref{fig:asymptotic_scaling}, plotted as $Q^2 |F_\pi(Q^2)|$ versus $Q^2$. This range illustrates the transition between non-perturbative and perturbative QCD regimes, a region targeted by experiments such as JLab E12-19-006. We benchmark the neural network's prediction against four theoretical bounds. 
\begin{itemize}
    \item[--]  Curve A represents the full, non-perturbative prediction derived from the Dyson-Schwinger equation (DSE) framework~\cite{Chang:2013nia}, which accounts for the dressed-quark mass function and dynamical chiral symmetry breaking (DCSB) across all energy scales. It transitions smoothly from the low-$Q^2$ non-perturbative regime into the high-$Q^2$ asymptotic regime, serving as the theoretical benchmark for this analysis.

    \item[--] Curve B shows the empirical monopole form ($F_\pi(Q^2)=1/(1+Q^2/m_\rho^2)$), representing the classic VMD extension into high energies. While this parametrisation matches low-energy data, it acts as an upper bound that overestimates the non-perturbative QCD behaviour at higher $Q^2$. 
\end{itemize}
The remaining curves map the short-distance, pQCD limits using the leading-order, leading-twist formula, $Q^2 F_\pi(Q^2)\approx 8\pi \alpha_s(Q^2)f_\pi^2\omega_\varphi^2$, with $\omega_\varphi=\int_0^1 dx\, \varphi_\pi(x)/(3x)$. 
\begin{itemize}
    \item[--] Curve C is the pQCD prediction obtained using a simplified, semi-analytical pion distribution amplitude (DA), $\phi_\pi(x;Q^2 = 4 \,\text{GeV}^2)\approx x^\varrho (1-x)^\varrho \Gamma[2(\varrho+1)]/\Gamma[(\varrho+1)]^2$, with $\varrho = 0.3$. Instead of computing the full non-perturbative machinery across all scales, this curve isolates how a realistic, flat DA profile, structured by DCSB at a standard hadronic initialisation scale of $4$ GeV$^2$, modulates the valence-quark gluon exchange kernel. Because it relies strictly on a hard-scattering factorisation framework, Curve C is initialised exclusively in the deep-spacelike domain ($Q^2 \gtrsim 7$ GeV$^2$) where the strong coupling constant $\alpha_s(Q^2)$ is sufficiently small to ensure perturbative convergence. 
    
    \item[--]
    Curve D shows the leading-order textbook formula for asymptotic pQCD using both the asymptotic DA, $\phi_{as}(x)=6 x (1-x)$ (the one we use in Appendix~\ref{sec:nlo}), and a standard running $\alpha_s(Q^2)$. This limit illustrates the expected behaviour if local, non-perturbative physics and dynamical mass generation vanished instantly outside the low-energy regime.
\end{itemize}
The PINN prediction (solid black curve) successfully bridges the non-perturbative and perturbative domains up to $Q^2 = 20\text{ GeV}^2$. At low-to-intermediate momentum transfers ($Q^2 \lesssim 4\text{ GeV}^2$), the network smoothly interpolates existing spacelike measurements while staying well below the empirical monopole bound (Curve B), which lacks the QCD-driven logarithmic suppression at higher momentum transfer.
\begin{figure}[!t]
    \centering
    \includegraphics[width=\columnwidth]{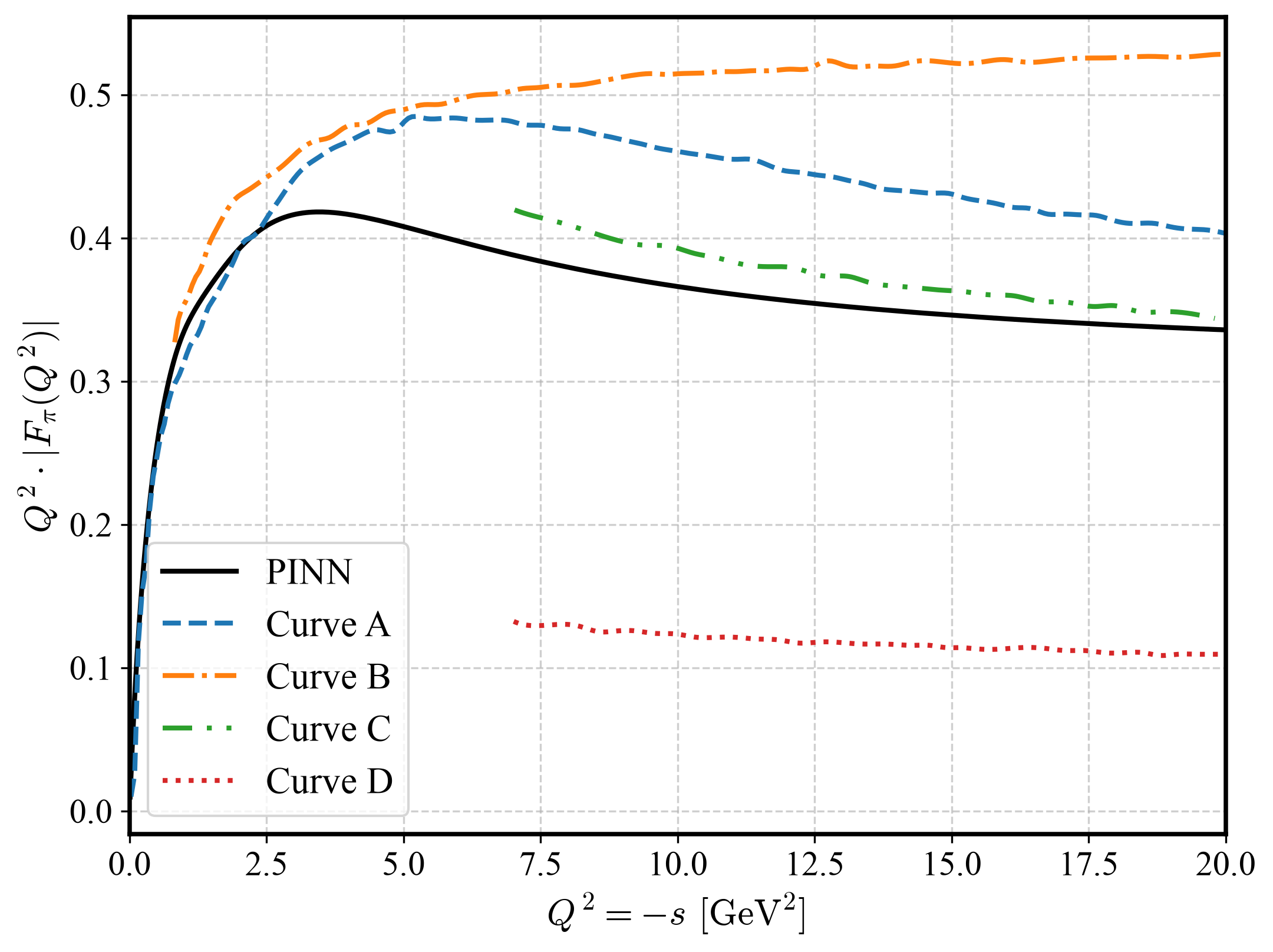}
    \caption{Asymptotic scaling of the pion form factor multiplied by the squared momentum transfer, $Q^2 |F_\pi(Q^2)|$, in the spacelike domain up to $Q^2 = -s = 20\text{ GeV}^2$. The solid black curve shows the PINN framework prediction. For comparison, Curve A (blue dashed) represents a QCD theoretical model bridging large and short distance scales; Curve B (orange dash-dotted) corresponds to a standard monopole parametrisation; and Curves C (green dash-dotted) and D (red dotted) illustrate high-energy short-distance quark-gluon approaches valid for $Q^2 \gtrsim 7\text{ GeV}^2$.\label{fig:asymptotic_scaling}}   
\end{figure}
%%%%%%%%%%%%%%%%%%%%%%%%%%%%%%%%%%%%%%%%%%%%%%%%%%%%%%%
In the high-$Q^2$ regime ($Q^2 \gtrsim 7\text{ GeV}^2$), the PINN prediction closely tracks the continuum DSE benchmark (Curve A) and approaches the broad-scale hard-scattering trajectory (Curve C). This close agreement stems from the fact that both the PINN (via its loss constraints) and the DSE/Curve C formulations preserve the non-perturbative dressing effects and the broad pion distribution amplitude generated by DCSB. Conversely, the PINN curve remains well above the asymptotic textbook limit (Curve D). Curve D assumes an immediate collapse to the asymptotic distribution amplitude $\phi_{as}(x)$, underestimating the persistent non-perturbative mass generation that survives at intermediate scales ($Q^2 \sim 10\text{--}20\text{ GeV}^2$). By enforcing asymptotic $1/s$ power-law falloff and Cauchy-Riemann analyticity without imposing a rigid parameterisation, the network organically favours the realistic, non-perturbative continuum QCD trajectory over rigid empirical extensions or premature asymptotic assumptions.

%%%%%%%%%%%%%%%%%%%%%%%%%%%%%%%%%%%%%%%%%%%%%%%%%%%%%%%
\begin{figure}[!t]
    \centering
    \includegraphics[width=\columnwidth]{"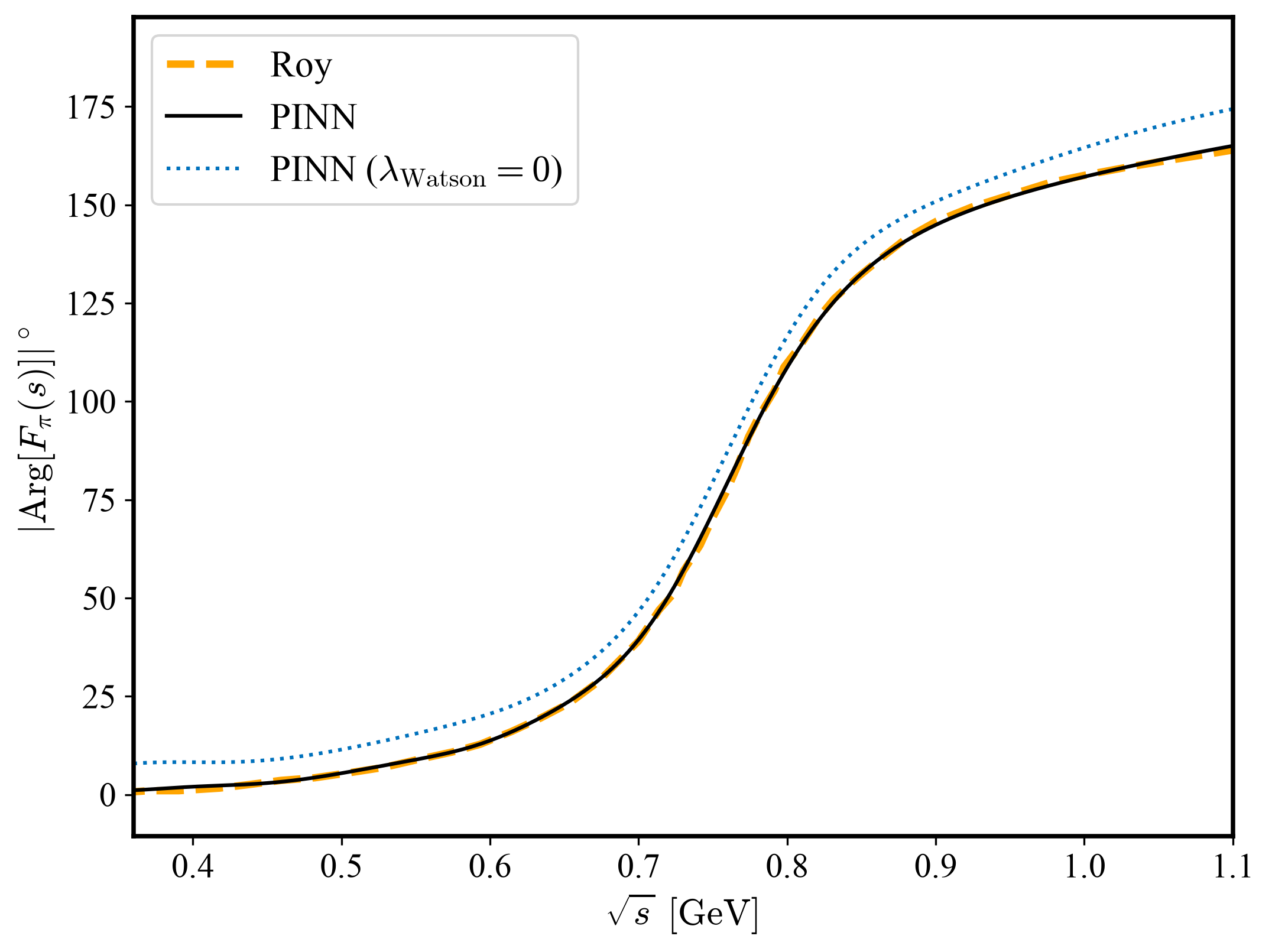"}
    \caption{The extracted phase of the form factor (black-solid) compared with the Roy equation data (orange-dashed) from Ref.~\cite{Colangelo:2018mtw}. The blue dotted curve shows the PINN prediction with Watson loss disabled ($\lambda_{\text{Watson}} = 0$), illustrating the systematic phase drift that results once the loss is removed. \label{fig:phase}}   
\end{figure}
%%%%%%%%%%%%%%%%%%%%%%%%%%%%%%%%%%%%%%%%%%%%%%%%%%%%%%%

%%%%%%%%%%%%%%%%%%%%%%%%%%%%%%%%%%%%%%%%%%%%%%%%%%%%%%%%%%%%%%%%%%%%%%

\subsection{The complex phase and Watson's theorem}\label{sec:res_phase}
\noindent
We evaluate the PINN's reconstruction of the complex phase of the pion electromagnetic form factor in Fig.~\ref{fig:phase}, to benchmark the prediction directly against the rigorous solutions of the Roy equations. The extracted phase tracks the Roy equation solution very closely across the full window, successfully capturing the localised derivative ($d\delta/ds$) that defines the physical decay width of the $\rho$ meson. The absence of any systematic offset between the two curves reflects the explicit Watson's theorem constraint ($\mathcal{L}_{\text{Watson}}$, Eq.~\ref{eq:watson-loss}) built into the training loss, which directly locks the phase of the isovector form factor to the elastic $\pi\pi$ scattering phase shift $\delta_1^1(s)$ in this kinematic region. To isolate the role of this constraint directly, Fig.~\ref{fig:phase} also shows the phase obtained when Watson's loss is switched off during training ($\lambda_{\text{Watson}} = 0$). The resulting curve (blue dotted) visibly departs from both the full PINN prediction and the Roy-equation solution, with the deviation growing with energy rather than remaining a static offset. This is the same effect quantified in the ablation study (Table~\ref{tab:result_abalation}). Analyticity and dispersion alone constrain the global holomorphic structure of $F_\pi(s)$, but they do not by themselves fix the functional form of the low-energy phase; that information must come from Watson's theorem. Without it, the network still produces a smooth, monotonic phase (since $\mathcal{L}_{\text{monotonicity}}$ is still active), but one that is systematically displaced from the physical $\pi\pi$ scattering phase shift.

Unlike a purely data-driven fit, the full PINN is architecturally constrained to satisfy elastic unitarity, rather than merely interpolating sparse experimental phase inputs. This agreement is nontrivial as standard Roy-equation solutions are anchored by sensitive low-energy boundary conditions such as the experimental $S$-wave scattering length ($a_0^0$) to which the PINN has no direct access. Instead, these global physics loss terms collectively constrain the low-energy phase; Watson’s theorem sets the target phase alignment in the elastic regime, while analyticity and dispersion enforce its smooth continuation across the entire kinematic domain.

Since the full PINN's phase reconstruction in Fig.~\ref{fig:phase} tracks the Roy-equation baseline with no discernible deviation (in contrast to the $\lambda_{\text{Watson}} = 0$ curve), we exploit the network's direct access to the second Riemann sheet to extract the $\rho(770)$ mass and width. A resonance is defined, process-independently, as a pole of the $\pi\pi$ scattering amplitude on the second (unphysical) Riemann sheet~\cite{Eden:1966dnq, BADALYAN198231}. We extract $m_\rho$ and $\Gamma_\rho$ by locating this pole directly, rather than relying on the conventional Breit-Wigner-motivated extraction from the real-axis phase (the 90° phase-crossing point and its local slope), which implicitly neglects any non-resonant background phase. Sheet II is reached from the network's Sheet-I output $F_\pi^I(s)$ via the elastic-unitarity continuation,
\[F_\pi^{II}(s) = F_\pi^I(s)\,e^{-2i\delta_1^1(s)},\]
valid below the inelastic threshold where Watson's theorem [Eq.~\ref{eq:watson_relation}] fixes the phase of $F_\pi^I(s)$ to the $P$-wave $\pi\pi$ phase shift $\delta_1^1(s)$. This is the standard technique used to extract light-meson poles from dispersive phase-shift solutions~\cite{Oller:2024lrk, Caprini:2005zr}. Writing the inverse amplitude as $D(s) \equiv 1/F_\pi^{II}(s) = S_1^1(s)/F_\pi^I(s)$, where $S_1^1(s)=e^{2i\delta_1^1(s)}$ is the elastic $\pi\pi$ $S$-matrix element, makes explicit that a zero of $D(s)$ is equivalent to a zero of $S_1^1(s)$ itself, i.e., a genuine $\pi\pi$ resonance pole, provided $F_\pi^I(s)$ contributes no spurious zeros of its own. This is supported by the zero-free analytic structure established in Fig.~\ref{fig:no_zeros}. Searching for the zero of $D(s)$ rather than the pole of $F_\pi^{II}(s)$ directly also converts an ill-posed numerical divergence into a well-behaved root-finding problem. 

A two-dimensional Newton-Raphson search over $s=s_R+is_I$, initialised near the $\rho$ peak, converges to
\[s_{\text{pole}} = (0.5756 - 0.1036\,i)\ \text{GeV}^2, \quad |D(s_{\text{pole}})|^2 \sim 2\times10^{-17},\]
confirming that the root is resolved to machine precision rather than sitting at a shallow local minimum. Using the standard parameterisation $\sqrt{s_{\text{pole}}} = m_\rho^{\text{pole}} - i \Gamma_\rho^{\text{pole}}/2$, we extract the conventional pole mass and width,
\begin{align*}
   & m_\rho^{\text{pole}} = \mathrm{Re}\sqrt{s_{\text{pole}}} = 761.72\pm 1.04\ \text{MeV}, \\
    &\Gamma_\rho^{\text{pole}} = -2\,\mathrm{Im}\sqrt{s_{\text{pole}}} = 135.99 \pm 1.20 \ \text{MeV}.
\end{align*}
As a cross-check, we also extract the resonance parameters from the conventional real-axis phase-crossing definition, $\delta_1^1(m_\rho^2)=\pi/2$, with the width inferred from the local phase slope, $\Gamma_\rho = \left[m_\rho\,(d\delta_1^1/ds)|_{s=m_\rho^2}\right]^{-1}$, giving $m_\rho = 773.36\pm 0.38\ \text{MeV}$ and $\Gamma_\rho \approx 149.59\pm 2.42\ \text{MeV}$.
This agrees with the pole extraction to within $\sim 1.5$\% in mass, while the width differs by about $10$\%, a statistically significant spread given the sub-MeV/few-MeV uncertainties on each extraction. This is expected as a slowly-varying non-resonant background phase barely shifts the steep $90^\circ$ crossing point, but contributes directly, and without suppression, to the local slope used to infer $\Gamma_\rho$, a bias the background-free pole extraction does not inherit. Both extractions are in the same ballpark as the PDG $T$-matrix pole, $\sqrt{s_{\text{pole}}}=(761\text{-}765)-i\,(71\text{-}74)\ \text{MeV}$~\cite{ParticleDataGroup:2024cfk}. Our pole mass agrees to within $1\%$, while our extracted pole width sits approximately $4\text{-}8\%$ below the PDG world average. This minor width deficit is a natural target for future multi-channel extensions of the framework. 

%%%%%%%%%%%%%%%%%%%%%%%%%%%%%%%%%%%%%%%%%%%%%%%%%%%%%
\subsection{Uncertainty quantification}\label{sec:stat_cali}
\noindent
To quantify the reliability of the predicted observables, such as the pion charge radius or the two-pion hadronic contribution to the muon anomalous magnetic moment, we decouple statistical data-driven fluctuations from systematic model choices by decomposing the total uncertainty budget into independent statistical ($\sigma_{\text{stat}}$) and calibration ($\sigma_{\text{cali}}$) components, quoting observables as $\hat{\mathcal O} \pm \sigma_{\text{stat}} \pm \sigma_{\text{cali}}$. The two are estimated from independent ensembles and combined in quadrature.

To estimate the statistical uncertainty ($\sigma_{\text{stat}}$), we employ a Monte Carlo bootstrapping procedure on the experimental dataset. We first generate pseudo-datasets by sampling individual measurements from Gaussian distributions defined by their reported central values and standard errors. For high-statistics datasets providing published covariance matrices, namely KLOE, BaBar, Belle, and CLEO, data points are drawn jointly from the corresponding multivariate Gaussian distribution to preserve bin-to-bin correlations; for all remaining datasets lacking published covariances, measurements are sampled independently. We construct $N = 500$ such pseudo-datasets and train an independent network on each realisation, where the bootstrap seed simultaneously governs the data sampling, network weight initialisations, and minibatch sequencing, so that $\sigma_{\text{stat}}$ reflects both data resampling and ordinary training variation. Final central values and their associated statistical uncertainties ($\sigma_{\text{stat}}$) for all observables are evaluated as the ensemble mean and standard deviation across these bootstrap realisations, respectively.

The loss weighting hyperparameters $\lambda_i$ carry no physical units and do not correspond to measurable observables or physical couplings; they are akin to calibration parameters in an experiment. As long as the network satisfies the corresponding constraint via $\mathcal L_i$, the physical observables will ideally satisfy $\partial\mathcal O/\partial\lambda_i=0$. In practice, the $\lambda_i$ act as numerical \emph{hyperparameters} that scale gradient step sizes during backpropagation, preventing any single constraint from dominating the optimisation landscape, and the residual dependence of $\mathcal O$ on $\lambda_i$ away from this ideal is precisely what we quantify below as $\sigma_{\text{cali}}$.

Because the relative weights ($\lambda_i$) represent modelling choices with no unique theoretical prescription, we systematically sample the ten-dimensional hyperparameter space to quantify the resulting variation in our observables. Instead of varying parameters one at a time, we vary each hyperparameter $\lambda$ over a factor of two, spanning the interval $[\lambda/2, 2\lambda]$, using a Latin Hypercube Sampling scheme. To guarantee uniform, space-filling coverage across the parameter domain that random sampling cannot achieve, the hyperparameter space is sampled at $N = 64$ points generated via a centred-discrepancy-optimised Latin hypercube. For each configuration, a neural network is trained using an identical adaptive learning-rate schedule and convergence criterion. The calibration uncertainty on each observable is then defined as the standard deviation across these $N = 64$ converged network configurations. Because each calibration trial represents a fully converged fit to the unperturbed central experimental data, the statistical and calibration ensembles are computationally independent and probe orthogonal sources of uncertainty. Consequently, we report the statistical and calibration uncertainties separately throughout, combining them in quadrature ($\sigma^2_{\text{tot}} = \sigma_{\text{stat}}^2 + \sigma_{\text{cali}}^2$) when determining the total uncertainty.

%%%%%%%%%%%%%%%%%%%%%%%%%%%%%%%%%%%%%%%%%%%%%%%%%%%%%
\subsection{Charge radius of the pion}
\noindent
The Taylor expansion of $F_\pi(s)$ at low momentum transfer ($s \to 0$) is defined in Eq.~\eqref{eq:charge_radius}. The physical consistency of our PINN framework is further benchmarked by extracting the pion squared charge radius, $\langle r_\pi^2 \rangle$, which is obtained from the first derivative of the form factor at zero momentum transfer: 
\begin{align}
    \langle r_\pi^2 \rangle = 3! \kappa^2 \left.\frac{dF_\pi(s)}{ds}\right|_{s=0},    
\end{align} 
where $\kappa = 0.197327$ GeV fm is the natural unit conversion factor from GeV$^{-1}$ to fm. As summarised in Table~\ref{tab:result_compare}, our results for the pion form-factor moments are in strong agreement with established theoretical and experimental values in the literature. For the squared charge radius $\langle r_\pi^2 \rangle$, our result of $0.435 \pm 0.008_{\text{stat}} \pm 0.007_{\text{cali}}$ fm$^2$ is consistent with the PDG world average~\cite{ParticleDataGroup:2024cfk}. Furthermore, it remains fully compatible with the lattice QCD prediction~\cite{Gao:2021xsm} at $0.67\sigma$, while reducing the statistical uncertainty relative to the lattice extraction by more than a factor of two. These precision gains are supported by the global nature of our framework, which constrains the derivative at $s=0$ using the full spacelike and timelike spectrum through dispersive integrals rather than local polynomial approximations.

For the higher-order curvature terms, our prediction for the fourth-order moment $\langle r_\pi^4 \rangle = 0.748 \pm 0.057_{\text{stat}} \pm 0.048_{\text{cali}}$ fm$^4$ shows excellent agreement ($0.39\sigma$) with two-loop ChPT~\cite{Bijnens:1998fm} and sits just above the range of the dispersive bounds obtained in Ref.~\cite{Ananthanarayan:2011xt}. Notably, our network constrains the uncertainty of $\langle r_\pi^4 \rangle$ to about two-thirds of that from two-loop ChPT.

Finally, our sixth-order moment $\langle r_\pi^6 \rangle = 4.01 \pm 0.95_{\text{stat}} \pm 0.90_{\text{cali}}$ fm$^6$ suggests a slightly enhanced low-energy curvature compared to earlier estimates, remaining consistent at the $0.8\sigma$ level with both the dispersive bounds of Ref.~\cite{Ananthanarayan:2011xt} and the estimate of Ref.~\cite{Truong:1998yx}. Such minor shifts are physically expected given that higher-order derivatives at $s=0$ act as sensitive probes of the complex-plane geometry, receiving non-local contributions from the timelike resonance spectrum and the high-$s$ asymptotic tail, both of which are dynamically coupled in our network via exact S-matrix constraints.

%%%%%%%%%%%%%%%%%%%%%%%%%%%%%%%%%%%%%%%%%%%%%%%%%%%%%%%
\begin{table}
\caption{\label{tab:result_compare}Comparison of the extracted pion charge radius squared $\langle r_\pi^2 \rangle$, fourth-order radius $\langle r_\pi^4 \rangle$, and sixth-order radius $\langle r_\pi^6 \rangle$ from our PINN framework against existing theoretical, lattice QCD, ChPT, and empirical estimates from the literature. All values are given in units of $\text{fm}^n$ (where $n = 2, 4, 6$). Errors for our work are presented in the form $\hat{\mathcal O} \pm \sigma_{\text{stat}} \pm \sigma_{\text{cali}}$ as defined in Sec.~\ref{sec:stat_cali}}
{\renewcommand{\arraystretch}{1.5}
\begin{tabular*}{\columnwidth}{@{\extracolsep{\fill}}l  rr}
\hline\hline
Observable & Previous estimation (in fm$^n$) & Our result (in fm$^n$)\\\hline
\multirow{2}{*}{$\langle r_\pi^2 \rangle$}
    &Lattice~\cite{Gao:2021xsm}: $0.42 \pm 0.02$  
        &\multirow{2}{*}{ $0.435 \pm 0.008\pm 0.007$
        } \\
    &PDG~\cite{ParticleDataGroup:2024cfk}: 
    $0.434\pm 0.005$&\\ \hline
\multirow{2}{*}{$\langle r_\pi^4 \rangle$}
    &$2$-loop ChPT~\cite{Bijnens:1998fm}: $0.7\pm 0.11$             
        &\multirow{2}{*}{$0.748 \pm 0.057\pm 0.048$
        } \\
    &Ref.~\cite{Ananthanarayan:2011xt}: $0.68- 0.72$&\\ \hline
\multirow{2}{*}{$\langle r_\pi^6 \rangle$}
    &Ref.~\cite{Ananthanarayan:2011xt}: $2.95-3.11$
        &\multirow{2}{*}{$4.01\pm0.95\pm 0.90$
        }\\
    &Ref.~\cite{Truong:1998yx}: $2.89 \pm 0.12$&\\
\hline\hline
\end{tabular*}}
\end{table}
%%%%%%%%%%%%%%%%%%%%%%%%%%%%%%%%%%%%%%%%%%%%%%%%%%%%%%%

%%%%%%%%%%%%%%%%%%%%%%%%%%%%%%%%%%%%%%%%%%%%%%%%%%%%%%%
\subsection{Implications for the muon anomalous magnetic moment}
\noindent
The hadronic contribution to the anomalous magnetic moment of the muon can be written as~\cite{Colangelo:2018mtw, Czarnecki:2001pv}:
\begin{equation}
a_\mu^{\pi \pi} = \left(\frac{\alpha_{em} m_\mu}{3 \pi}\right)^2 \int_{4m_\pi^2}^\infty ds \frac{K(s)}{s^2} R_{\pi \pi}(s),
\end{equation}
where $\alpha_{em} = e^2/4\pi$, is the fine structure constant, and $K(s)$ is the kernel function dominating at low energies:
\begin{align}
K(s) =& \frac{3s}{m_\mu^2}\Bigg[\frac{x^2}{2}(2-x^2)+\frac{(1+x^2)(1+x)^2}{x^2}\nonumber\\
&\ \times\left(\log(1+x)-x+\frac{x^2}{2}\right)+\frac{1+x}{1-x}x^2 \log x\Bigg],
\end{align}
with
\begin{align*}
x =& \frac{1-\sigma_\mu(s)}{1+\sigma_\mu(s)},~ ~\sigma_\mu(s) =\sqrt{1-\frac{4 m_\mu^2}{s}},    
\end{align*}
and $R_{\pi \pi}$ is related to the specific case of two-pion production ($e^+ e^- \to \pi^+ \pi^-$), and is directly proportional to the square of the pion form factor.
\begin{align*}
R_{\pi\pi}(s) = \frac{1}{4}\left(1-\frac{4 m_\pi^2}{s}\right)^{3/2}|F_\pi(s)|^2.
\end{align*}
In the presence of FSR, the above equation is modified to~\cite{Colangelo:2018mtw}:
\begin{align}
R^{\text{FSR}}_{\pi\pi}(s) =\left[1+\frac{\alpha_{em}}{\pi}\mathfrak f(s)\right]R_{\pi\pi}(s)    \label{eq:rad_cor}
\end{align}
where $\mathfrak f(s)$ is defined in Eq.~\eqref{eq:fsr}. In Table~\ref{tab:amupipi_compare}, we present our results for $a_\mu^{\pi\pi}$ integrated over three characteristic energy domains alongside prominent estimations in the literature. In the low-energy region below $0.63\,\text{GeV}$, our result ($130.66 \pm 0.67_{\text{stat}} \pm 1.20_{\text{cali}}$) is consistent with the dispersive evaluation of Colangelo et al.~\cite{Colangelo:2018mtw} within $0.14\sigma$ (evaluating total uncertainties in quadrature), while differing from Ananthanarayan et al.~\cite{Ananthanarayan:2016mns} by $1.67\sigma$. For the primary $\rho$-resonance region $[2m_\pi, 1.0\,\text{GeV}]$, our value of $496.37 \pm 2.08_{\text{stat}} \pm 1.70_{\text{cali}}$ is consistent with Colangelo et al.~\cite{Colangelo:2018mtw}, differing by $0.37\sigma$. Finally, extending the integration across the entire spectrum $[2m_\pi, \infty]$ yields $a_\mu^{\pi\pi} \times 10^{10} = 506.48 \pm 2.02_{\text{stat}} \pm 1.70_{\text{cali}}$, compatible with the benchmark evaluation of the Muon $g-2$ Theory Initiative, Ref.~\cite{Aoyama:2020ynm}, within $0.34\sigma$, and $1.45\sigma$ lower than Ref.~\cite{Davier:2017zfy}. 

Our physics-informed framework achieves a refined total uncertainty ($\pm 2.64$ in quadrature) compared to traditional evaluations (from $\pm 3.29$ to $\pm3.8$). Standard data-driven integrations directly propagate point-to-point experimental noise and local dataset discrepancies through the $K(s)/s^2$ integration kernel. In contrast, the PINN provides a globally consistent analytic framework by enforcing Cauchy-Riemann analyticity, Watson's final-state interaction theorem, and $1/s$ pQCD asymptotic falloff within the conformal $z$-plane. It suppresses unphysical local fluctuations and resolves sharp $\rho$-$\omega$ interference features without relying on rigid functional parameterisations. The resulting representation of $F_\pi(s)$ reproduces established low-energy moments and $a_\mu^{\pi\pi}$ integrals while establishing a robust, physics-constrained baseline for hadronic vacuum polarisation.

%%%%%%%%%%%%%%%%%%%%%%%%%%%%%%%%%%%%%%%%%%%%%%%%%%%%%%%
\begin{table}
\caption{\label{tab:amupipi_compare} Comparison of the two-pion hadronic vacuum polarisation contribution to the muon anomalous magnetic moment, $a_\mu^{\pi\pi} \times 10^{10}$, integrated over various centre-of-mass energy intervals ($\sqrt{s}$ in GeV). Our PINN estimations are benchmarked against existing dispersive and data-driven determinations from the literature.  Errors for our work are presented in the form $\hat{\mathcal O} \pm \sigma_{\text{stat}} \pm \sigma_{\text{cali}}$ as defined in Sec.~\ref{sec:stat_cali}.}
{\renewcommand{\arraystretch}{1.5}
\begin{tabular*}{\columnwidth}{@{\extracolsep{\fill}}l  rr}
\hline\hline
Energy range & \multirow{2}{*}{Previous estimations} &\multirow{2}{*}{Our estimation} \\
(GeV) & & \\ \hline
\multirow{2}{*}{$[2m_\pi,0.63]$}
    &$133.258 \pm 0.723$~\cite{Ananthanarayan:2016mns}  
        &\multirow{2}{*}{$130.66\pm 0.67\pm1.20$
        } \\
    &
    $130.47\pm 0.64$~\cite{Colangelo:2018mtw}&\\ \hline
$[2m_\pi,1.0]$
    &$495.0\pm 2.58$~\cite{Colangelo:2018mtw}             &
    $496.37\pm2.08\pm1.70$\\
    \hline
\multirow{2}{*}{$[2m_\pi, \infty]$}
    &$513.2\pm3.8$~\cite{Davier:2017zfy}
        &\multirow{2}{*}{$506.48\pm 2.02\pm1.70$} \\
    & $507.9 \pm 3.29$~\cite{Aoyama:2020ynm} &\\
    %https://arxiv.org/pdf/2006.04822 (page 37 and 40)
    %Eq.~3.49 of Ref.~\cite{Aoyama:2020ynm}
\hline\hline
\end{tabular*}}
\end{table}
%%%%%%%%%%%%%%%%%%%%%%%%%%%%%%%%%%%%%%%%%%%%%%%%%%%%%%%

%%%%%%%%%%%%%%%%%%%%%%%%%%%%%%%%%%%%%%%%%%%%%%%%%%%%%%%
\begin{table*}[t]
\caption{\label{tab:result_abalation} Performance comparison of pion form factor moments $\langle r_\pi^{2n}\rangle$, integrated HVP contributions $a_{\mu}^{\pi\pi} \times 10^{10}$, and P-wave phase-shift deviations $|\Delta\delta_1^1(s)|$ from the Roy equation at $s_0 = (0.8\text{ GeV})^2$ and $s_1 = (1.15\text{ GeV})^2$ across dataset variations (Sets A-D) and physics-loss ablation studies on Set A.}
{\renewcommand{\arraystretch}{1.3}
\begin{tabular*}{\textwidth}{@{\extracolsep{\fill}}l
rrrrrrrr}
\hline\hline
\multirow{2}{*}{Configuration} & 
\multirow{2}{*}{$\langle r_\pi^2 \rangle$ ($\text{fm}^2$)} & 
\multirow{2}{*}{$\langle r_\pi^4 \rangle$ ($\text{fm}^4$)} & 
\multirow{2}{*}{$\langle r_\pi^6 \rangle$ ($\text{fm}^6$)} & 
\multicolumn{3}{c}{$a_\mu^{\pi\pi} \times 10^{10}$} & 
\multicolumn{2}{c}{$|\Delta\delta_1^1(s)|$ in degrees} \\ 
\cline{5-7} \cline{8-9}
& & & & $[2m_\pi,0.63]$ & $[2m_\pi,1.0]$ & $[2m_\pi,\infty]$ & $s=s_0$ & $s=s_1$ \\ 
\hline
\multicolumn{9}{l}{\textbf{Data ablation (baseline sets)}} \\ 
\hline
Set A (no CMD-3) & $\mathbf{0.435}$ & $\mathbf{0.748}$ & $\mathbf{4.01}$ & $\mathbf{130.66}$ & $\mathbf{496.37}$ & $\mathbf{506.48}$ & $\mathbf{0.25}$ & $\mathbf{1.72}$ \\ 
Set B (CMD-3 only) & $0.447$ & $0.738$ & $3.335$ & $137.71$ & $520.08$ & $532.27$ & $0.08$ & $5.08$ \\ 
Set C$_{\rm raw}$ ($\tau$-decay) & $0.421$ & $0.665$ & $2.709$ & $131.36$ & $502.34$ & $512.79$ & $1.15$ & $6.65$ \\ 
Set C$_{\rm R_{IB}}$ ($\tau$-decay) & $0.437$ & $0.817$ & $5.458$ & $127.40$ & $492.51$ & $502.83$ & $0.19$ & $6.28$ \\ 
Set D (all combined) & $0.421$ & $0.655$ & $2.515$ & $132.69$ & $504.76$ & $515.30$ & $0.42$ & $3.46$ \\ 
$\lambda_{\text{spacelike}} = 0$ (Set A) & $0.436$ & $0.786$ & $4.486$ & $132.67$& $496.12$ & $505.68$ & $0.19$ & $0.60$ \\ 
$\lambda_{\text{timelike}} = 0$ (Set A) & $0.392$ & $2.075$ & $31.560$ & $17.72$ & $20.18$ & $20.83$ & $55.99$ & $116.20$ \\ \hline
\multicolumn{9}{l}{\textbf{Physics-loss ablation (trained on Set A)}} \\ \hline
Set A baseline & $\mathbf{0.435}$ & $\mathbf{0.748}$ & $\mathbf{4.01}$ & $\mathbf{130.66}$ & $\mathbf{496.37}$ & $\mathbf{506.48}$ & $\mathbf{0.25}$ & $\mathbf{1.72}$  \\ 
$\lambda_{\text{analyticity}} = 0$ & $0.431$ & $0.700$ & $2.956$ & $132.76$ & $496.07$ & $506.48$ & $0.90$ & $2.47$ \\ 
$\lambda_{\text{dispersion}} = 0$ & $0.448$ & $0.836$ & $5.464$ & $130.65$ & $498.28$ & $508.28$ & $0.36$ & $2.91$ \\ 
$\lambda_{\text{moment}} = 0$ & $0.436$ & $0.705$ & $12.701$ & $130.77$ & $497.37$ & $507.33$ & $0.19$ &$1.90$\\ 
$\lambda_{\text{positivity}} = 0$ & $0.436$ & $0.758$ & $4.053$ & $130.04$ & $497.36$ & $507.54$ & $0.34$ & $2.72$ \\ 
$\lambda_{\text{Watson}} = 0$ & $0.546$ & $1.337$ & $11.188$ & $135.01$ & $502.78$ & $512.73$ & $7.70$ & $11.45$ \\ 
$\lambda_{\text{monotonicity}} = 0$ & $0.438$ & $0.759$ & $4.119$ & $129.80$ & $496.68$ & $506.67$ & $0.77$ & $2.20$ \\ 
$\lambda_{\text{pQCD}} = 0$ & $0.435$ & $0.745$ & $3.986$ & $130.82$ & $495.01$ & $505.11$ &$0.28$ & $3.75$ \\ 
$\lambda_{\text{t-asymptotic}} = 0$ & $0.431$ & $0.729$ & $3.746$ & $130.49$ & $495.46$ & $505.76$ & $0.42$ & $5.19$ \\ 
%$\lambda_{\text{mixing}} = 0$ & $0.462$ & $0.922$ & $6.842$ & $129.59$ & $499.40$ & $509.29$ & $1.04$ & $0.26$ \\ 
$\lambda_{phys} = 0$ & $0.445$ & $-0.640$ & $-57.210$ & $131.74$ & $498.50$ & $509.27$ & $43.60$ & $57.20$ \\ 
\hline
Pad\'e approximation & $0.453$ & $0.834$ & $4.758$ & $130.14$ & $492.52$ & $502.30$ & $0.48$ & $3.33$ \\
\hline\hline
\end{tabular*}} 
\end{table*}
%%%%%%%%%%%%%%%%%%%%%%%%%%%%%%%%%%%%%%%%%%%%%%%%%%%%%%%

%%%%%%%%%%%%%%%%%%%%%%%%%%%%%%%%%%%%%%%%%%%%%%%%%%%%%%%
 \begin{figure*}
    \captionsetup[subfigure]{labelformat=empty}
    \centering
    \subfloat[{\bf (a)} Set C$_{\rm raw}$]{\includegraphics[width=0.49\textwidth]{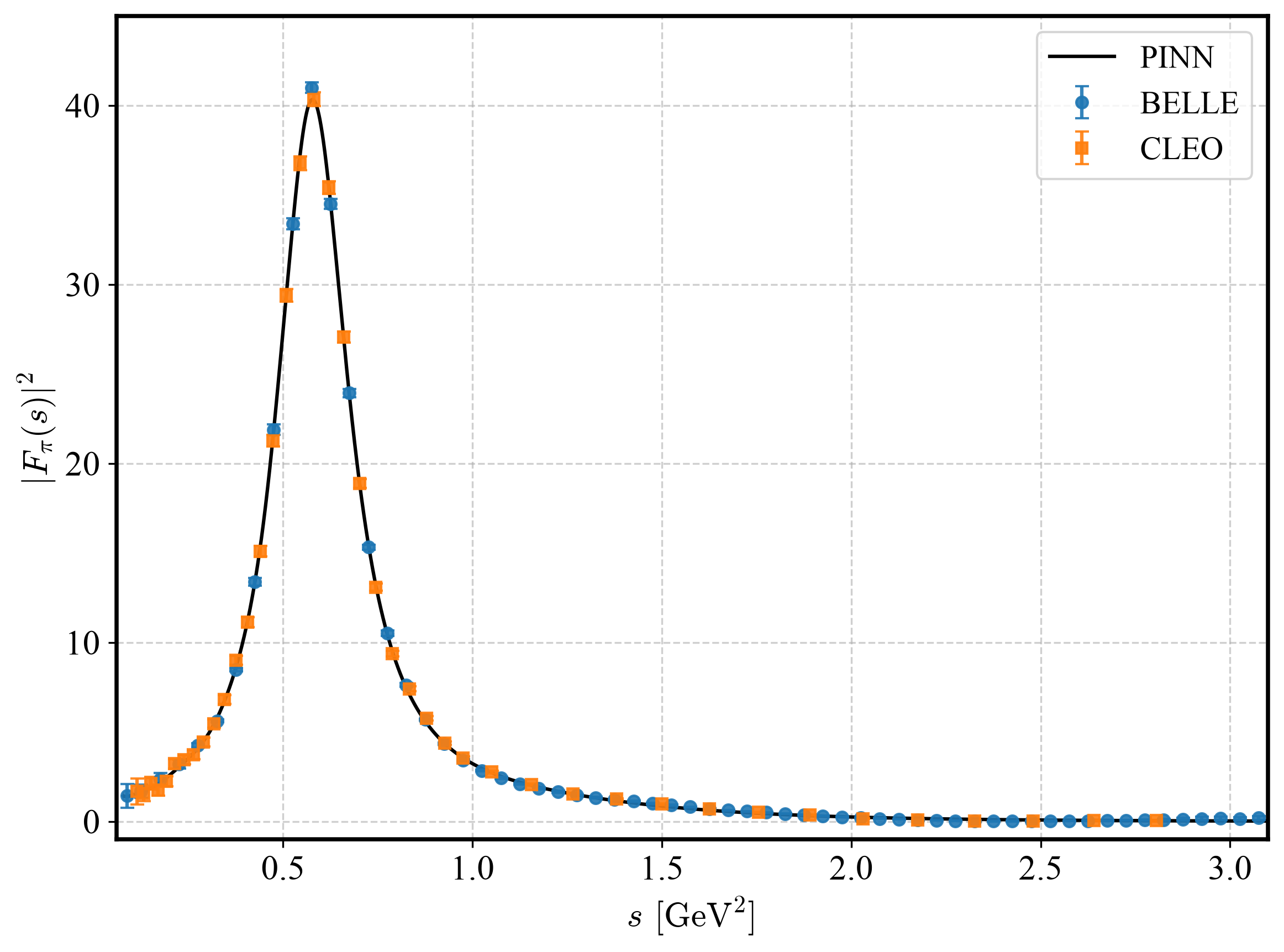}\label{fig:raw_tau}}\hfill
    \subfloat[{\bf (b)} Set C$_{R_{\rm IB}}$]
    {\includegraphics[width=0.49\textwidth]{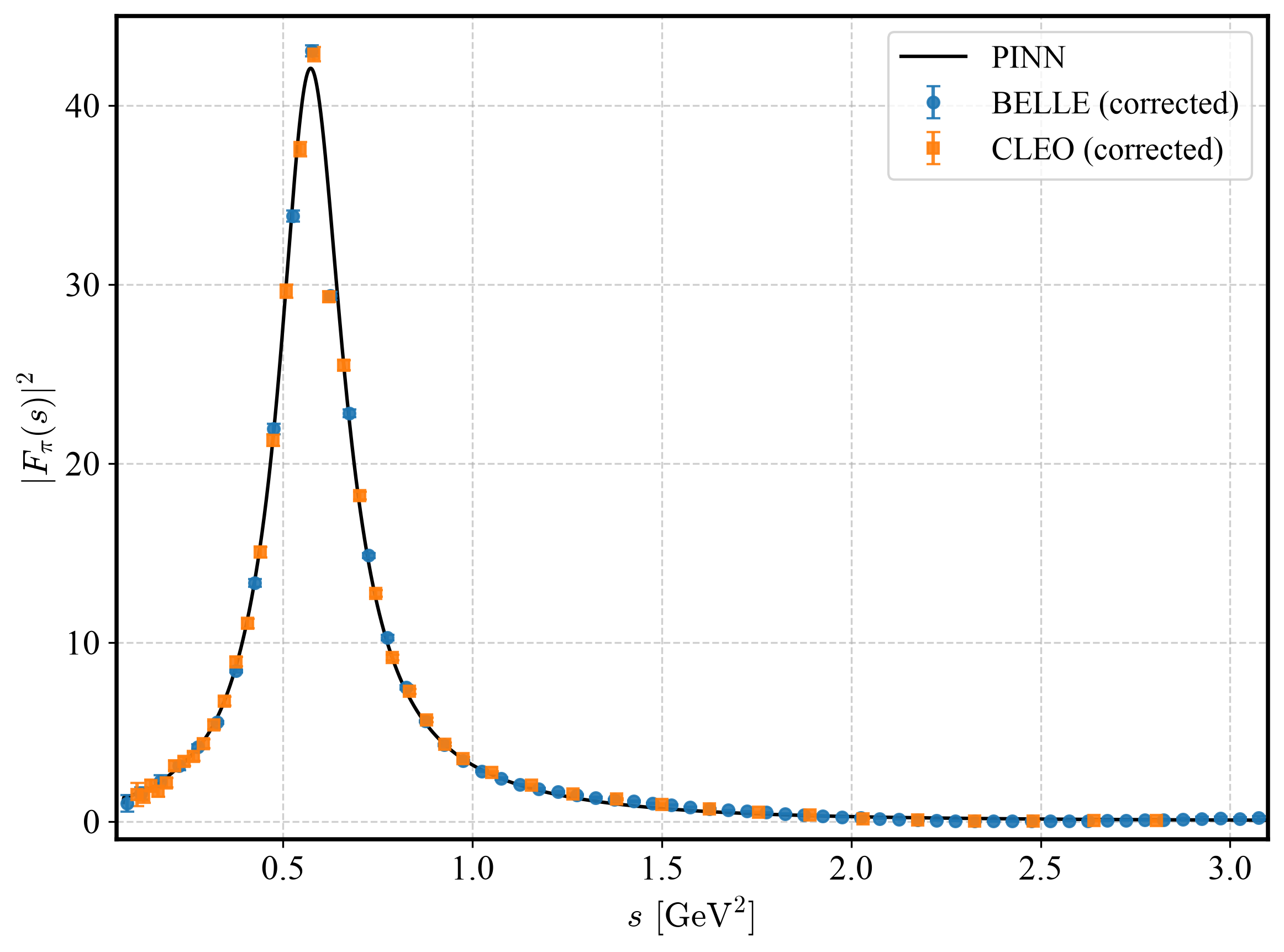}\label{fig:epem_like}}
    \caption{Reconstruction of $|F_{\pi}(s)|^2$ using $\tau$-decay data in the timelike region under two distinct treatment schemes: (a) Set $C_{\rm raw}$, where raw $\tau$-decay data are directly fitted while isolating the pure $I=1$ isovector form factor $F_{\pi}^{I=1}(s)$ by setting the $\rho-\omega$ mixing parameter to zero ($\alpha_{\rho-\omega} = 0$); and (b) Set $C_{R_{\rm IB}}$, where raw $\tau$ data are explicitly pre-corrected for isospin-breaking effects using the point-by-point $R_{IB}(s)$ factor before training. Both panels display the resulting PINN predictions overlaid against experimental data points from BELLE and CLEO. \label{fig:timeliketau}}
\end{figure*}
%%%%%%%%%%%%%%%%%%%%%%%%%%%%%%%%%%%%%%%%%%%%%%%%%%%%%%%

%%%%%%%%%%%%%%%%%%%%%%%%%%%%%%%%%%%%%%%%%%%%%%%%%%%%%%%
\subsection{Hyperparameter loss weighting and structural ablation}\label{sec:ablation}
\noindent
We probe the sensitivity of the loss functionals through ablation ($\lambda_i = 0$ vs. $\lambda_i \neq 0$). By selectively setting individual weights to zero, we isolate and test the structural necessity of each theoretical principle (such as Watson’s theorem, pQCD asymptotics, or dispersion relations) by observing their effect on various physical quantities (see Table~\ref{tab:result_abalation}). Physical observables such as $a_\mu^{\pi\pi}$ do not directly enter the loss functional but instead emerge a posteriori as outcomes evaluated via dispersion integrals over $F_\pi(s)$; nonetheless, as the results below show, they remain highly sensitive to which constraints and datasets are active during training. \medskip

\noindent{\bf Significance of first-principles learning:} Our results indicate that data-driven constraints alone guarantee only localised numerical interpolation, whereas the global uniqueness and structural integrity of the pion form factor rely on the physics losses. The numbers in Table~\ref{tab:result_abalation} reveal a clear hierarchy among the losses based on their role in shaping $F_\pi(s)$. Watson's elastic unitarity theorem emerges as the single most vital physics loss in the framework. Removing this constraint ($\lambda_{\text{Watson}} = 0$) destabilises the network: $\langle r_\pi^2 \rangle$, $\langle r_\pi^4 \rangle$, and $\langle r_\pi^6 \rangle$ inflate by about $25\%$, $80\%$, and $179\%$, respectively, and the phase shifts show marked deviations of $7.70^\circ$ at $s_0$ and $11.45^\circ$ at $s_1$. Watson's theorem ties the phase of $F_\pi(s)$ directly to elastic $\pi\pi$ scattering phase shifts below the inelastic threshold; without it, the network loses its directional phase alignment, causing distortions in the imaginary spectrum, and spoils the low-energy slope. 

On the other hand, disabling the Cauchy-Riemann analyticity constraint ($\lambda_{\text{analyticity}} = 0$) or the explicit dispersion relation ($\lambda_{\text{dispersion}} = 0$) causes modest shifts in the HVP integral and low-energy moments, but introduces growing phase-shift errors at higher energies ($\vert{}\Delta\delta_1^1(s_1)\vert{} \approx 2.5^\circ\text{--}2.9^\circ$). Intuitively, global analyticity and Cauchy integrals enforce smooth connectivity across the complex plane; removing them allows local gradient noise to unbind the real and imaginary parts of the form factor above the $\rho(770)$ peak. Similar phase deviations are seen by setting $\lambda_{\text{monotonicity}} = 0$, which requires the P-wave phase shift to increase monotonically with energy ($\partial \delta_1^1 / \partial s > 0$) in the elastic domain. Disabling monotonicity enhances the phase error at higher energies ($\sim2.20^\circ$ at $s_1$). Monotonicity prevents numerical backtracking and unphysical phase oscillations, particularly near resonance thresholds where phase shifts change rapidly.

The moment constraint ($\lambda_{\text{moment}} = 0$) plays a different, more localised role. Omitting this curvature sum rule causes $\langle r_\pi^6\rangle$ to explode to about $12.7$ fm$^6$ (more than triple the baseline value), while leaving $a_\mu^{\pi\pi}$ virtually untouched. This happens because neural networks normally suffer from spectral bias, a tendency to smooth out high-frequency curvature near $s = 0$. The moment loss acts as an essential low-energy anchor penalising unphysical local flattening.

The network output shows only a mild sensitivity towards spectral positivity ($\lambda_{\text{positivity}} = 0$), as the dense experimental data in the elastic resonance region naturally force $\text{Im}\,F_\pi(s) \geqslant 0$. The loss mainly serves as an unphysical-sign regulator that prevents the spectral function from becoming negative in sparse kinematic regions or high-energy tails. Similarly, the asymptotic constraints ($\lambda_{\text{pQCD}} = 0$ and $\lambda_{t\text{-asymptotic}} = 0$) only act as high-energy boundary controls; disabling them allows the network's high-energy tail to drift freely, degrading the scattering phase at $s_1$ up to about $5^\circ$ without affecting the other quantities significantly.

If we switch off all physics constraints ($\lambda_{phys} = 0$), we get a purely data-driven model. However, it suffers a complete physical breakdown even though it easily fits discrete timelike data points by introducing high-frequency numeric oscillations. While these local wrinkles average out in broad spectral integrals like $a_\mu^{\pi\pi}$, taking sequential derivatives magnifies them catastrophically, forcing $\langle r_\pi^4 \rangle$ and $\langle r_\pi^6 \rangle$ to collapse to unphysical negative numbers and destroying the phase alignment completely. Incorporating physics losses is therefore indispensable: they act as non-local regularisers that guarantee holomorphy, enforce unitarity, and yield reliable physical derivatives that unconstrained interpolations cannot achieve.\medskip

\noindent{\bf The tension in the data:} We test the roles of different datasets similarly through ablations. Removing spacelike data ($\lambda_{\text{spacelike}} = 0$) causes minor degradation in the low-energy moments, while removing timelike data ($\lambda_{\text{timelike}} = 0$) causes a catastrophic collapse in $a_\mu^{\pi\pi}$ and a massive phase shift ($\sim116^\circ$), reflecting the dominant role of the timelike dataset near the $\rho(770)$ resonance in fixing the form factor. 

Training on various data subsets exposes the well-known tension across experiments. Training exclusively on CMD-3 data along with the spacelike data (Set B) pushes $a_\mu^{\pi\pi}[2m_\pi, \infty]$ to $ 532.27 \times 10^{-10}$ from our baseline value (Set A: $506.48 \times 10^{-10}$). This correlates with the higher cross-section normalisation reported by CMD-3 near the $\rho(770)$ resonance peak, which also induces a noticeable phase shift ($5.08^\circ$ at $s_1 = (1.15\text{ GeV})^2$) relative to the Roy equation solution. This indicates that the CMD-3 data creates a subtle structural tension with elastic unitarity and dispersion constraints at higher energies, forcing the model to distort its phase slope to accommodate the excess spectral weight. 

We next evaluate the network's sensitivity to $\tau$-decay data via Set C (Fig.~\ref{fig:timeliketau}). In Set $C_{\text{raw}}$, the PINN isolates the pure $I=1$ form factor $F_{\pi}^{I=1}(s)$ natively by setting $\alpha_{\rho-\omega} = 0$ (using the channel switch), bypassing external pre-corrections. Retaining this unsuppressed $I=1$ normalisation yields $a_{\mu}^{\pi\pi}[2m_{\pi},\infty] = 512.79 \times 10^{-10}$, demonstrating that embedding channel-aware parameterisations within the PINN provides a self-consistent alternative to model-dependent point-by-point factors. Conversely, in Set $C_{R_{\text{IB}}}$, the point-by-point $R_{\text{IB}}(s)$ pre-corrections [Eq.~\ref{eq:ribmixing}] disrupt local derivative structures and distort the resonance curvature. Enforcing S-matrix analyticity and dispersion relations causes the physics loss to smooth through these gradient mismatches, slightly underfitting the peak near $s \approx 0.6\text{ GeV}^2$ (Fig.~\ref{fig:epem_like}) and pulling $a_{\mu}^{\pi\pi}$ down to $502.83 \times 10^{-10}$. Due to reduced high-mass precision in $\tau$ data, both Set C configurations exhibit larger phase shift mismatches ($|\Delta\delta_1^1(s_1)| = 6.65^\circ$ for Set $C_{\text{raw}}$ and $6.28^\circ$ for Set $C_{R_{\text{IB}}}$) compared to Set A ($1.72^\circ$), underscoring the reliance on high-energy $e^+e^-$ data for phase consistency above $1\text{ GeV}^2$.

In Set D, which synthesises all $e^+e^-$ and raw $\tau$ data streams via the conditional channel switch, the PINN acts as an analytical mediator. The network integrates the conflicting CMD-3 and $e^+e^-$/$\tau$ normalizations to yield a balanced HVP contribution of $a_{\mu}^{\pi\pi}[2m_{\pi},\infty] = 515.30 \times 10^{-10}$. This upward shift relative to Set A ($506.48 \times 10^{-10}$) reflects how raw $\tau$ normalisation and CMD-3 data jointly pull the global fit toward higher spectral density. Despite the added input tension, Set D maintains strong physical self-consistency, achieving a high-energy phase shift deviation of $|\Delta\delta_{1}^{1}(s_{1})| = 3.46^\circ$. This is an improvement over standalone $\tau$ sets ($> 6^\circ$), demonstrating how $e^+e^-$ high-energy precision stabilises S-matrix analyticity while accommodating raw $\tau$ dynamics natively.

%%%%%%%%%%%%%%%%%%%%%%%%%%%%%%%%%%%%%%%%%%%%%%%%%%%%%%%
\begin{figure*}
    \captionsetup[subfigure]{labelformat=empty}
    \centering
    \subfloat[{\bf (a)} $\left|\,|F^{\rm PINN}_\pi| - |F^{\text{Pad\'e}}_\pi|\,\right|$]{\includegraphics[width=\textwidth]{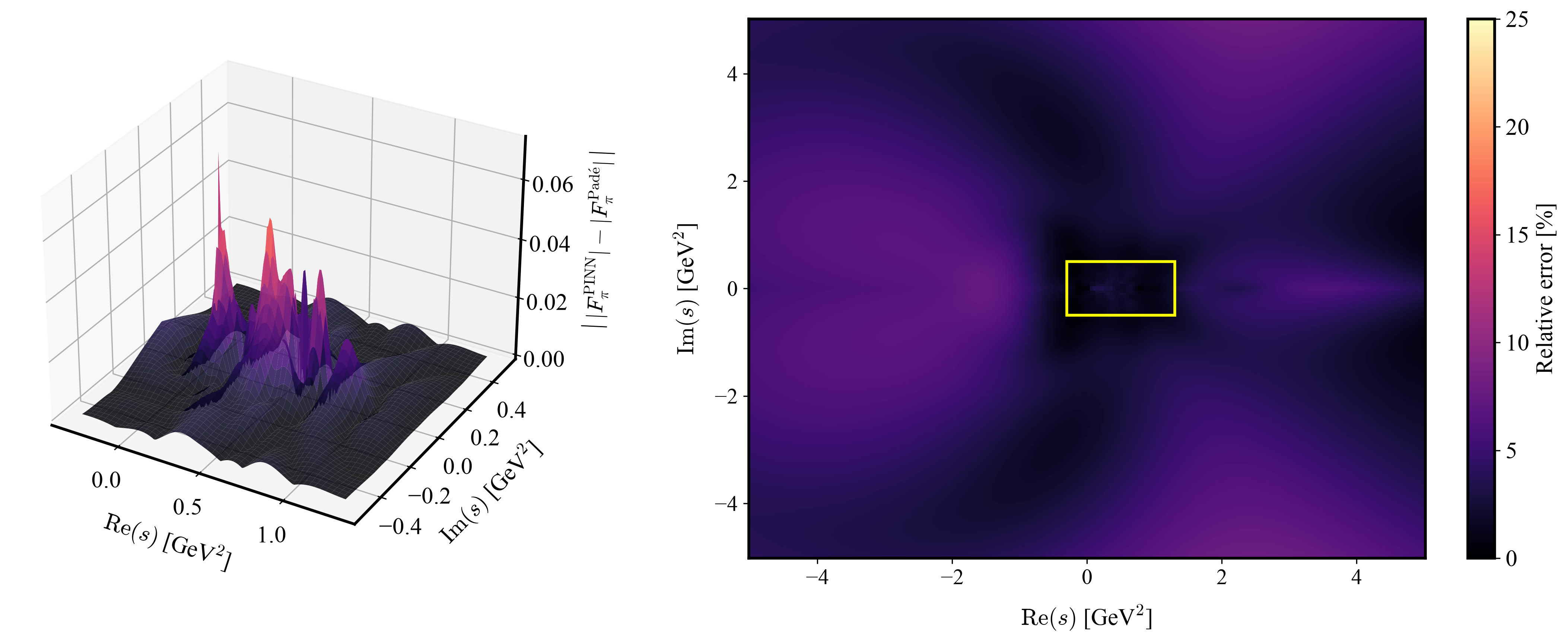}\label{fig:pade_complex_plane_3d1}}\\

        \subfloat[{\bf (b)} $\left|\,\text{Arg}(F^{\rm PINN}_\pi) - \text{Arg}(F^{\text{ Pad\'e}}_\pi)\,\right|$ ]{\includegraphics[width=\textwidth]{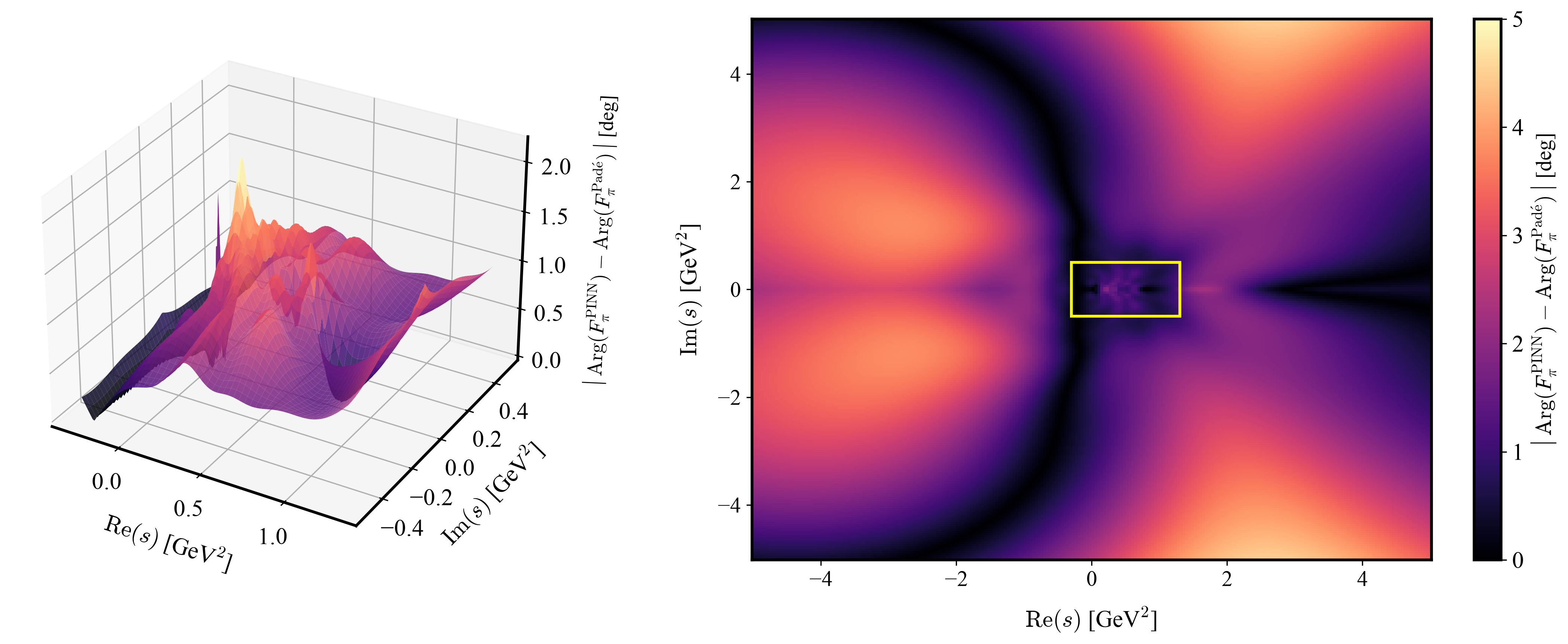}\label{fig:pade_complex_plane_3d2}}
    \caption{Discrepancy between the full PINN prediction and its $[4,4]$ Pad\'e approximation across the complex $s$-plane. (a) Absolute difference in magnitude, $\big||F_\pi^{\rm PINN}(s)| - |F_\pi^{\text{Pad\'e}}(s)|\big|$: the 3D surface (left) zooms into the region marked by the yellow box in the 2D heat map (right), which shows the relative error over the full sampled domain. (b) Corresponding phase difference, $\big|\mathrm{Arg}(F_\pi^{\rm PINN}) - \mathrm{Arg}(F_\pi^{\text{ Pad\'e}})\big|$ in degrees, with the same left/right correspondence. }
\end{figure*}
%%%%%%%%%%%%%%%%%%%%%%%%%%%%%%%%%%%%%%%%%%%%%%%%%%%%%%%
\begin{figure}[!t]
    \centering
    \includegraphics[width=\columnwidth]{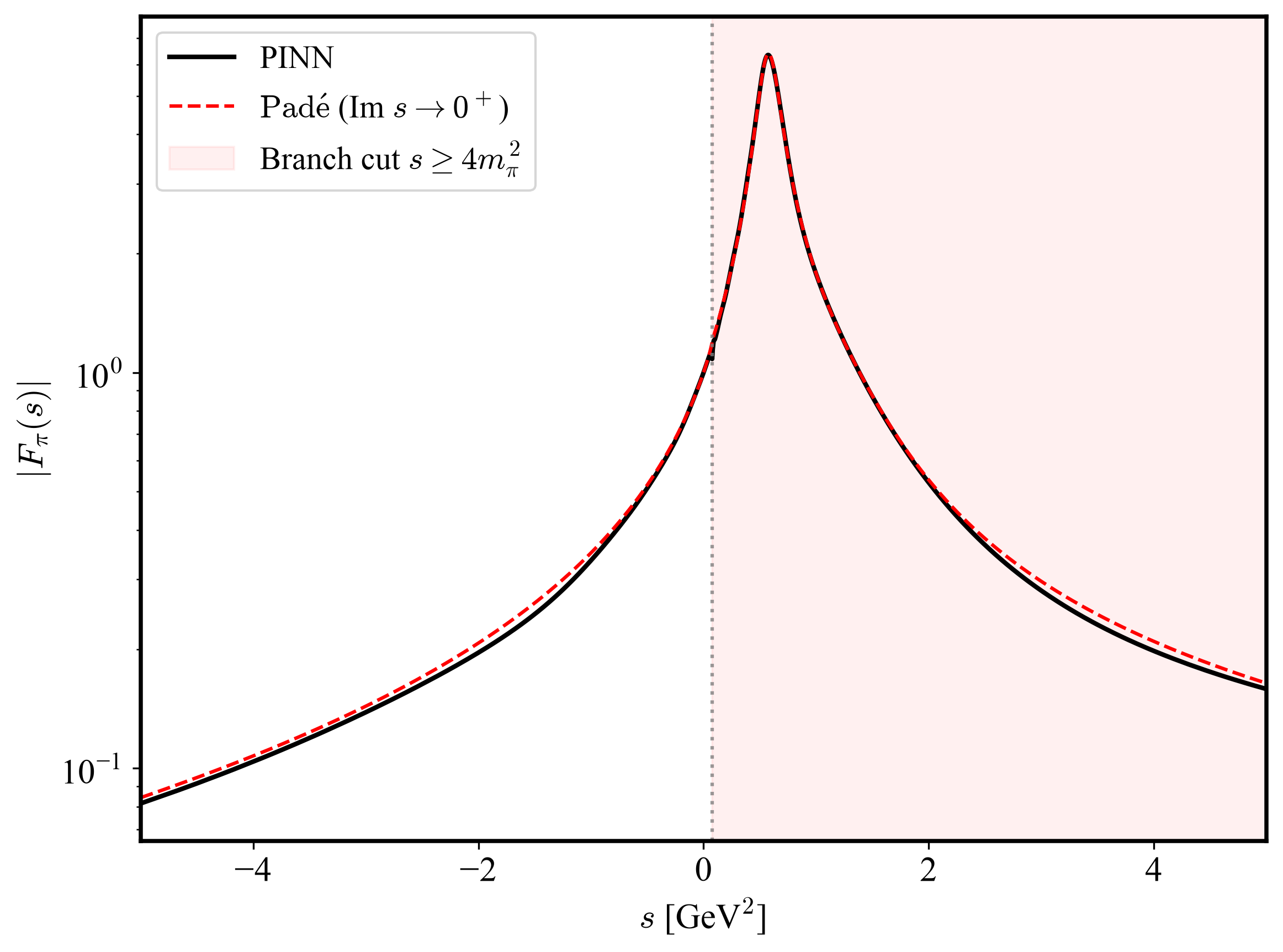}
    \caption{Comparison of $|F_\pi(s)|$ from the full PINN (solid black) and the [4/4] Pad\'e approximant (red dashed) along the physical axis, across the spacelike and timelike regions. The two curves are visually indistinguishable through the $\rho(770)$ peak. \label{fig:paderealline}}   
\end{figure}
%%%%%%%%%%%%%%%%%%%%%%%%%%%%%%%%%%%%%%%%%%%%%%%%%%%%%%%

%%%%%%%%%%%%%%%%%%%%%%%%%%%%%%%%%%%%%%%%%%%%%%%%%%%%%%%
\section{Closed-form approximation of the form factor}\label{sec:pade}
\noindent
Although the PINN gives a numerically stable, model-independent form factor, its output exists only as a trained network, not as an expression that can be quoted or used in other calculations. We find a compact conformal Pad\'e fraction in $z=z(s)$ that closely follows the PINN prediction near the real axis and the $\rho(770)$ pole, useful for phenomenological calculations and further model building,
\begin{equation}\label{eq:pade1}
F_\pi(s) \;=\;
\frac{1 + a_1 z + a_2 z^2 + a_3 z^3 + a_4 z^4}
     {\left(1 - \dfrac{z}{z_\rho}\right)
      \left(1 - \dfrac{z}{z^*_\rho}\right)
      \left(1 + b_1 z + b_2 z^2\right)},
\end{equation}
with best-fit coefficients
\begin{align*}
a_1 &= -3.527647, & a_2 &= 4.755645, \\ 
a_3 &= -2.937709, & a_4 &= 0.707660, \\
b_1 &= -1.554048, & b_2 &= 0.758975, \\ 
z_\rho &= 0.790353  -0.728033\,i.&&
\end{align*}
The complex-conjugate pole pair $z_\rho, \, z_\rho^*$ is factored out explicitly as $(1-z/z_\rho)(1-z/z_\rho^*)$ in the denominator. This structural choice exactly enforces the Schwarz reflection property of the network output (any pole at $z_\rho$ is automatically accompanied by its mirror pole at $z_\rho^*$). The remaining quadratic $(1+b_1z+b_2z^2)$ supplies the additional rational structure needed to match the network beyond the leading resonance behaviour. The fit converges to all four denominator roots lying outside the unit disk $(|z|> 1)$. This means the approximant is automatically holomorphic throughout $|z|\le 1$, consistent with the analyticity enforced by $\mathcal L_{\text{analyticity}}$ during training, without this having been enforced by hand.

We show in Figs.~\ref{fig:pade_complex_plane_3d1},~\ref{fig:pade_complex_plane_3d2} and \ref{fig:paderealline}, how well Eq.~\eqref{eq:pade1} reproduces the full network prediction. Each panel shows the discrepancy between the full PINN prediction and its Pad\'e approximation over the complex $s$-plane as both a 3D surface (left) and a 2D heat map (right). Fig.~\ref{fig:pade_complex_plane_3d1} shows the absolute difference in magnitude, $\big||F_\pi^{\rm PINN}(s)| - |F_\pi^{\text{Pad\'e}}(s)|\big|$, with the heat map colour-scaled to the relative error. Fig.~\ref{fig:pade_complex_plane_3d2} shows the corresponding difference in phase, $\big|\mathrm{Arg}(F_\pi^{\rm PINN}) - \mathrm{Arg}(F_\pi^{\text{Pad\'e}})\big|$, in degrees. Both quantities remain small near the real axis and around the $\rho(770)$ pole, which is the region of interest for most physical calculations; the discrepancy grows to several per cent in magnitude and a few degrees in phase further into the complex plane.

We next show a direct comparison of the form factor along the physical kinematic line in Fig.~\ref{fig:paderealline}. The form factor $|F_\pi(s)|$ from the PINN (solid black) and from the Pad\'e approximant (red dashed) are overlaid on a logarithmic scale across both the spacelike ($s<0$) and timelike ($s>4m_\pi^2$, shaded pink) regions. The two curves are visually indistinguishable through the $\rho(770)$ peak, with a small, growing separation appearing only at the largest timelike $s$ shown, confirming that the closed-form expression faithfully reproduces the network's behaviour where it matters most for the extracted observables.

Because the second sheet is reached from the first via $z_{II}(s) = 1/z_I(s)$, the fitted pole $z_\rho$ maps directly onto the $\rho(770)$ pole position by evaluating the inverse conformal map at $z_{II, \rho} = 1/z_{I,\rho}$,
\begin{align*}
    s_\rho =&\ 4m_\pi^2\left[1 - \left(\frac{1+z_\rho^{-1}}{1-z_\rho^{-1}}\right)^2\right]=\ 0.5737 - 0.1065\,i~\text{GeV}^2.
\end{align*}
This gives us,
\begin{align*}
    \sqrt{s_\rho} =& \left(760.7 - \tfrac{i}{2}140.1\right)~\text{MeV}\\
    \Rightarrow\; m_\rho =&\ 760.7~\text{MeV},\quad \Gamma_\rho = 140.1~\text{MeV}.
\end{align*}
This is an independent extraction, obtained purely from the algebraic pole of the closed-form fit rather than from the Newton–Raphson root search on the full network output (Sec.~\ref{sec:res_phase}). It agrees with the direct extraction to within $\sim$ 0.1\% in mass and $\sim$ 3\% in width. The residual difference reflects the finite order of the Pad\'e truncation rather than any inconsistency in the underlying network, and confirms that a low-order rational function already captures the resonance analytic structure faithfully.

To systematically select the appropriate order for this closed-form surrogate, we refit conformal Pad\'e approximants of orders $[1,1]$ through $[6,6]$ directly to the PINN output and track the mean relative deviation on the timelike cut. The deviation drops rapidly from 124\% at $[1,1]$ to 7.1\% at $[2,2]$ and 3.4\% at $[4,4]$, after which it saturates (3.4\% at $[5,5]$, 3.2\% at $[6,6]$); the spacelike axis and complex plane exhibit a matching plateau (2.8–3.3\% and 2.8–3.1\%, respectively). The extracted resonance parameters likewise stabilise from $[4,4]$ onward. We therefore select the $[4/4]$ order as the minimal approximant that fully saturates the achievable accuracy, balancing mathematical compactness against numerical fidelity.

Evaluating the physical observables directly from the $[4/4]$ Pad\'e approximant yields $a_\mu^{\pi\pi}$ values in close agreement with the continuous PINN baseline (the last line of Table~\ref{tab:result_abalation}), though the higher moments $\langle r_\pi^n\rangle$ and the phase deviation at $s_1$ show somewhat larger residual shifts. This indicates that while the PINN solves the ill-posed inverse problem by mapping noisy, multi-channel experimental data into a smooth, physically constrained representation across the real energy axis, the Pad\'e approximant acts as an analytic tool that compresses this continuous representation into a rational form, granting direct access to unphysical-sheet resonance poles and low-energy derivative expansions. The slight shift in the low-energy moments ($\langle r_{\pi}^{n}\rangle$) reflects the minimal-order truncation redistributing far-UV asymptotic tails, confirming that the Pad\'e fit serves as an analytic bridge from the PINN's continuous representation to exact S-matrix properties, rather than as a solver in its own right.

%%%%%%%%%%%%%%%%%%%%%%%%%%%%%%%%%%%%%%%%%%%%%%%%%%%%%%%
\section{Discussion and future directions}\label{sec:discussion}
\noindent
We have presented a Physics-Informed Neural Network embedded in a conformal $z$-plane that extracts the pion electromagnetic form factor $F_\pi(s)$ across the full spacelike and timelike spectrum directly from first principles, rather than from a chosen functional template. The conformal mapping resolves a structural pathology we identify in standard neural-network training on unbounded kinematic variables. Embedding charge normalisation, analyticity, dispersion relations, and Watson's theorem into the loss functional then yields a form factor free of complex zeros without this being separately imposed, along with precise extractions of the pion charge radius, the $\rho(770)$ pole, and $a_\mu^{\pi\pi}$ (Secs.~\ref{sec:nn}--\ref{sec:results}).

A few points are worth discussing further. First, our default result excludes CMD-3: the network's global holomorphy and dispersion constraints fit CMD-3 less comfortably than BaBar and $\tau$-decay data, evidence bearing on the ongoing tension between these measurements, though the tension is not resolved by this framework alone. Second, the extracted pole width sits $4$--$8\%$ below the PDG world average, a small but consistent deficit likely tied to the elastic-unitarity approximation used to reach the second sheet; resolving it will need an explicit treatment of inelastic channels. Third, the zero-free structure and the CMD-3/BaBar preference both emerge without being directly targeted by any loss term, which is the kind of result most useful for adjudicating between datasets, but it should be read as evidence from one particular architecture and constraint set, not as an independent physical proof.

Extending this framework to multi-channel form factors (such as $\pi\pi \rightarrow K\bar{K}$ and $\pi\pi \rightarrow \omega\pi^0$) would let it capture inelastic effects beyond $1\,\text{GeV}^2$ directly, addressing the width deficit noted above. A unified treatment of momentum-dependent $\rho$--$\omega$ mixing and $\pi^\pm$--$\pi^0$ mass splitting within the network would further reconcile the $\tau$-decay and $e^+e^-$ datasets without relying on external isospin-breaking corrections. Finally, adapting the conformal PINN scheme to the pion transition form factor $F_{\pi^0\gamma^*\gamma^*}$ would provide a model-independent baseline for the dominant remaining theoretical uncertainty in Hadronic Light-by-Light scattering for $(g-2)_\mu$.

%%%%%%%%%%%%%%%%%%%%%%%%%%%%%%%%%%%%%%%%%%%%%%%%%%%%%%%
\section*{Supplementary material}
\noindent
The code, model weights \& plotting scripts are all available at our \href{https://github.com/mitra-subhadip/PINNingPion}{GitHub} repository.
%%%%%%%%%%%%%%%%%%%%%%%%%%%%%%%%%%%%%%%%%%%%%%%%%%%%%%%%%%%%%
\begin{acknowledgments}
\noindent
We thank B. Ananthanarayan for insightful discussions and valuable suggestions during the development of this work.
\end{acknowledgments}
%%%%%%%%%%%%%%%%%%%%%%%%%%%%%%%%%%%%%%%%%%%%%%%%%%%%%%%
\appendix
%%%%%%%%%%%%%%%%%%%%%%%%%%%%%%%%%%%%%%%%%%%%%%%%%%%%%%%
\section{Proofs of theorems}\label{sec:thoremproofs}
\setcounter{theorem}{0}
\noindent
We prove the NTK rank theorem from Section~\ref{sec:nn}.
\begin{theorem}[Full rank of the empirical NTK]
Let $z_1,\dots,z_N$ be distinct points with $|z_i|\leqslant 1$, and let $f(z;\theta)$ be a feedforward neural network with a real analytic, non-polynomial activation function $\sigma$ and a total of $\mathcal N$ hidden neurons across all layers. If $\mathcal N\geqslant N$, then for a generic choice of parameters $\theta$, the parameter Jacobian $J_z\in\mathbb R^{N\times P}$ has row rank $N$, and the empirical NTK $\Theta_z=J_zJ_z^T/P$ is positive definite: $\lambda_{\min}(\Theta_z)>0$, where $P$ is the total number of parameters, satisfying $P\geqslant\mathcal N$ since every neuron carries at least one parameter.
\end{theorem}
\begin{proof}
It is enough to show that there exists at least one parameter choice for which $J_z$ has row rank $N$. Since the entries of $J_z$ are analytic functions of the parameters, the vanishing of every $N\times N$ minor defines an analytic variety. If one such minor is not identically zero, then its zero set has an empty interior and measure zero. Hence, full row rank holds generically. Since $\Theta_z=J_zJ_z^T/P$ is positive definite exactly when $J_z$ has full row rank, positive-definiteness of the NTK then also holds generically.

Let us first consider a network with a single hidden layer of $\mathcal N$ neurons, $f(z;\theta)=\sum_{k=1}^{\mathcal N} v_k\sigma(a_kz+b_k)$. If we fix $\{a_k,b_k\}_{k=1}^{\mathcal N}$ at generic values, so that the $\mathcal N$ numbers $a_kz_i+b_k$ are pairwise distinct across neurons, the classical interpolation result for non-polynomial activations~\cite{Pinkus:1999} tells us that, once $\mathcal N\geqslant N$, the outer weights $\{v_k\}$ alone can be chosen so that the feature matrix $\Phi_{ik}=\sigma(a_kz_i+b_k)$ has rank $N$. Since $\partial_{v_k}f(z_i;\theta)=\Phi_{ik}$, the Jacobian $J_z$ contains $\Phi$ as a submatrix, so $J_z$ already has row rank $N$ at this parameter choice, establishing the theorem for a network with a single hidden layer.

To extend this to a network of any depth, the key observation is that adding more layers cannot reduce the freedom available to the network. Once one layer of neurons already gives the network some rank, adding a further layer that does nothing to the output leaves the rank untouched. Since its own neurons are also free to add further independent directions on top, the rank only increases. We now make this precise.

Let us express a network with $L$ hidden layers as $h^{(0)}(z)=z$ and $h^{(l)}(z)=\sigma(W^{(l)}h^{(l-1)}(z)+b^{(l)})$ for $l=1,\dots,L$, with output $f(z;\theta)=\sum_{k=1}^{\mathcal N_L} v_k\, h^{(L)}_k(z)$, where $\mathcal N_L$ is the width of the last hidden layer. Let $\mathcal N_{\leqslant l}$ denote the number of neurons in the first $l$ layers, and let $J_z^{(\leqslant l)}$ be the Jacobian restricted to the parameters of these layers. Suppose that, at some parameter choice, $J_z^{(\leqslant l)}$ has rank $r=\min(\mathcal N_{\leqslant l},N)$. If we now attach layer $l+1$ and tune its weights close to the values that make $h^{(l+1)}(z)$ simply reproduce $h^{(l)}(z)$, the rank already achieved is preserved, since the network at this point is only a small, generic deformation of the one we started with, and rank cannot drop under such a deformation. If $r$ is still less than $N$, we can then apply the single-layer argument above once more, this time to layer $l+1$'s own weights and to the inputs $h^{(l)}(z_i)$, which are generically distinct; this supplies further independent columns, raising the rank up to $\min(\mathcal N_{\leqslant l+1},N)$. Starting from $l=1$ and repeating this step up to $l=L$, we find that a generic choice of $\theta$ gives $J_z$ rank $\min(\mathcal N,N)$, where $\mathcal N=\mathcal N_{\leqslant L}$ is the total neuron count of the network. Once $\mathcal N\geqslant N$, this is exactly the row rank $N$ we set out to prove.
\end{proof}

\noindent
The condition $\mathcal N\geqslant N$ above only makes use of each neuron's outer weight, and is therefore a conservative one: for the sine activation we used, a single neuron in fact carries more than one useful direction. We illustrate this with a simple example.

Let us take a single neuron, $f(z;\theta)=v\sin(az+b)$, and three distinct points $z_1,z_2,z_3$. Its three parameters contribute the columns $\partial_v f=\sin(az_i+b)$, $\partial_b f=v\cos(az_i+b)$, and $\partial_a f=v\,z_i\cos(az_i+b)$. Note that $a$ cannot be zero here, since then $\partial_v f$ and $\partial_b f$ both become constant across $i$, making the resulting $3\times3$ determinant vanish identically. If instead we take $a$ small but non-zero and $b=0$, expanding the determinant in powers of $a$ gives, to leading order,
\begin{equation}
D(a,0) = -\frac{v}{3}\,a^3 \det\!\begin{pmatrix} z_1 & 1 & z_1^3\\ z_2 & 1 & z_2^3\\ z_3 & 1 & z_3^3 \end{pmatrix} + \mathcal O(a^5).
\end{equation}
The determinant on the right is not identically zero, so it can only vanish for special, symmetric choices of the three points. For any other choice, $D(a,0)\neq0$, and the single neuron already supplies rank $3$ from its three parameters. The same mechanism should extend to larger $N$. Our network used $\mathcal N=512$ neurons spread over four hidden layers, comfortably above the size of any collocation batch $N$ used in training, so the conservative bound already covers our case.

\begin{theorem}[Boundedness of Hessian in $\mathcal{Z}$-space]
For $f(z;\theta) = \sum_{k=1}^{K} v_k \sigma(a_k z+b_k)$ with $|z|\leqslant 1$, assume that the activation satisfies,
\begin{equation}
    |\sigma(u)|\leqslant C_0,\ 
    |\sigma'(u)|\leqslant C_1,\ 
    |\sigma''(u)|\leqslant C_2,
\end{equation}
for all $u\in\mathbb{R}$. Assume further that along the optimisation trajectory, $|v_k|\leqslant V$, and that the data labels are bounded $|y_i|\leqslant Y$. Then the Hessian of the MSE loss is uniformly bounded $\|H\|_2 \leq C_H < \infty$, where $C_H$ is independent of the original physical coordinate $s$.
\end{theorem}

\begin{proof}
Following Eq.~\eqref{eq:hessian}, the Hessian is
\begin{equation}
    H = \frac{1}{N} \sum_{i=1}^{N} \nabla_\theta f(z_i;\theta) \nabla_\theta f(z_i;\theta)^T + \frac{1}{N}
    \sum_{i=1}^{N} r_i(\theta)\nabla_\theta^2 f(z_i;\theta).
\end{equation}
We bound both terms separately. First, using $|z|\leqslant 1$ and the boundedness assumptions on the parameter derivatives, we get,
\begin{align}
    \left|\partial_{v_k} f\right| \leqslant C_0,\ 
    \left|\partial_{a_k} f\right| \leqslant VC_1,\ \text{and}\ 
    \left|\partial_{b_k} f\right| \leqslant VC_1.
\end{align}
Therefore there exists a constant $M_1<\infty$ such that
\begin{equation}
    \|\nabla_\theta f(z_i;\theta)\|_2 \leqslant M_1,
\end{equation}
for every data point $i$. Now, because $|\sigma''|\leqslant C_2$ and the input and the parameters are bounded, every second parameter derivative of $f$ is bounded by a constant depending only on $V, C_1, C_2$. Hence, there exists $M_2<\infty$ such that
\begin{equation}
    \|\nabla_\theta^2 f(z_i;\theta)\|_2 \leqslant M_2.
\end{equation}
It remains to bound the residuals. Since
\begin{equation}
    |f(z_i;\theta)| \leqslant \sum_{k=1}^{K} |v_k|\,|\sigma(a_k z_i+b_k)| \leqslant KVC_0,
\end{equation}
and $|y_i|\leqslant Y$, we have
\begin{equation}
    |r_i(\theta)| \leqslant KVC_0+Y\equiv R_{\max}.
\end{equation}
Therefore,
\begin{align}
    \|H\|_2 \leqslant&\ \frac{1}{N} \sum_{i=1}^{N} \|\nabla_\theta f(z_i;\theta)\|_2^2 + \frac{1}{N}
    \sum_{i=1}^{N} |r_i(\theta)| \|\nabla_\theta^2 f(z_i;\theta)\|_2\nonumber\\
    \leqslant&\ M_1^2 + R_{\max}M_2 \equiv C_H < \infty.
\end{align}
The bound depends on the network-size and weight bounds, but not on the original physical coordinate $s$.
\end{proof}

\section{The NLO pQCD constraint for the asymptotic loss functional}\label{sec:nlo}
\noindent
To anchor the PINN in the deep Euclidean region ($Q^2 = -s \rightarrow \infty$), we enforce the NLO pQCD prediction for the pion electromagnetic form factor. Following Melić et al.~\cite{Melic:1998qr,*Melic:1999mx}, the NLO form factor can be written as:
\begin{align}
    F_{\pi}(Q^2, \mu_R^2, \mu_F^2) =&\ F_{\pi}^{(0)}(Q^2, \mu_R^2, \mu_F^2) + F_{\pi}^{(1a)}(Q^2, \mu_R^2, \mu_F^2)\nonumber\\
    &\ + F_{\pi}^{(1b)}(Q^2, \mu_R^2, \mu_F^2),    
\end{align}
where $\mu_F$ is the factorisation scale, $F_{\pi}^{(0)}$ is the leading-order (LO) form factor, $F_{\pi}^{(1a)}$, the NLO correction to the hard-scattering amplitude, and $F_{\pi}^{(1b)}$, the NLO evolutional correction to the pion DA. To extract a closed-form analytic constraint for our loss functional, we choose the asymptotic DA~\cite{Melic:1998qr,*Melic:1999mx} to model the pion, $\phi_{as}(x) = 6x(1-x)$, where $x$ is the longitudinal momentum fraction carried by the quark. For the asymptotic DA, the $\mu_F$ dependence disappears and the NLO evolutional correction $F_{\pi}^{(1b)}$ becomes negligible ($\sim 1\%$). We can thus safely omit it when establishing a deep-spacelike boundary. Under this asymptotic assumption, the LO and the net NLO contributions simplify to:
\begin{align}
Q^2 F_{\pi}^{(0)}(Q^2, \mu_R^2) =&\ 8\pi f_{\pi}^2 \alpha_S(\mu_R^2),\\
Q^2 F_{\pi}^{(1)}(Q^2, \mu_R^2) \approx&\ 8 f_{\pi}^2 \alpha_S^2(\mu_R^2) \left[ \frac{\beta_0}{4} \ln\left(\frac{\mu_R^2}{Q^2}\right) + 6.83 - \frac{n_f}{12} \right],
\end{align}
where $n_f$ denotes the number of active flavours. Meli\'c et al. set $n_f=3$ to estimate their final results. However, since our spacelike data extends to $Q^2$ values above the charm and bottom thresholds, we adopt the five-flavour scheme, setting $n_f = 5$ and $\beta_0 = 11 - 2n_f/3 = 23/3$, and evaluate $\alpha_S(\mu_R^2)$ using the corresponding five-flavour renormalisation group evolution. Because the $\mu_F$ dependence vanishes for the asymptotic DA, this scheme change is implemented directly with no additional corrections.

Summing the LO and NLO hard-scattering contributions, we obtain the scale-dependent NLO prediction:
\begin{align}
F_{\pi}(Q^2, \mu_R^2) \approx \frac{8\pi f_{\pi}^2 \alpha_S(\mu_R^2)}{Q^2}& \Bigg[ 1 + \frac{\alpha_S(\mu_R^2)}{\pi} \Big( \frac{\beta_0}{4} \ln\left(\frac{\mu_R^2}{Q^2}\right) \nonumber\\
&\ + 6.41 \Big) \Bigg].
\end{align}
As pointed out by Meli\'c et al., setting $\mu_R^2 = Q^2$ is physically unsuited because the essential virtualities of the particles in the parton subprocess are considerably smaller than the overall momentum transfer $Q^2$. Following the Brodsky-Lepage-Mackenzie procedure adopted by Meli\'c et al., setting $\mu_R^2$ to the characteristic virtualities of the parton subprocess yields the physical renormalisation scale $\mu_R^2=\mu_{\rm BLM}^2 \approx Q^2/21$ for the asymptotic distribution amplitude. 

%%%%%%%%%%%%%%%%%%%%%%%%%%%%%%%%%%%%%%%%%%%%%%%%%%%%%%%%%%%%%%%%%%%%%%%%%%%%
\bibliography{reference}
%%%%%%%%%%%%%%%%%%%%%%%%%%%%%%%%%%%%%%%%%%%%%%%%%%%%%%%%%%%%%%%%%%%%%%%%%%%%
\end{document}